%% file: 05-08-2026.tex
\documentclass[12pt,a4paper]{article}
\input{preamble}
\begin{document}

\title{A characterization of the von Neumann-Morgenstern stable set in matching markets\thanks{Thanks \ldots. We acknowledge financial support
from UNSL through grants 032016 and 030320, from Consejo Nacional
de Investigaciones Cient\'{\i}ficas y T\'{e}cnicas (CONICET) through grant
PIP 112-200801-00655, and from Agencia Nacional de Promoción Cient\'ifica y Tecnológica through grant PICT 2017-2355. Mauricio Lucero Quevedo also acknowledges financial support from the predoctoral contract PREP2023-001667, funded by the Spanish Ministry of Science, Innovation and Universities and co-funded by the European Social Fund Plus (ESF+), within the framework of research project PID2023-150472NB-I00.}}


\author{Mauricio Lucero Quevedo\thanks{Departament de Matemàtica Econòmica, Financera i Actuarial, Facultat d’Economia i Empresa, Universitat de Barcelona, Barcelona, Spain. Emails: \texttt{luceroqam@ub.edu} (M. Lucero Quevedo).}   \and Paola Manasero\thanks{Instituto de Matem\'{a}tica Aplicada San Luis (UNSL-CONICET) and Departamento de Matemática, Universidad Nacional de San
Luis, San Luis, Argentina. Emails: \texttt{pbmanasero@unsl.edu.ar} (P. Manasero), \texttt{paneme@unsl.edu.ar} (P. Neme) and \texttt{joviedo12@gmail.com} (J. Oviedo).} \and Pablo Neme\samethanks[3] \and Jorge Oviedo\samethanks[3]}

\date{\today}
\maketitle

\begin{abstract}

This paper characterizes the von Neumann--Morgenstern (vNM) stable set in one-to-one matching markets with strict preferences. We prove that the vNM stable set is unique and coincides with the core of an explicitly constructed reduced preference profile.

The reduction combines two complementary sources of information: the lattice structure of core matchings and the efficiency-adjusted deferred acceptance mechanism. A key ingredient is a generalization of the classical Decomposition Lemma, which extends the cycle decomposition of core matchings to arbitrary internally stable sets. These structural results identify the pairs that remain relevant for dominance-based stability and lead to the reduced profile.

Our characterization shows that a cooperative solution concept defined through dominance relations can be computed by solving a classical core problem on an explicitly constructed reduced preference profile. As a consequence, it yields a constructive procedure for computing the unique vNM stable set and a bounded path-to-stability result, showing that every individually rational matching reaches the core in at most three dominance steps.

\bigskip

\bigskip


\end{abstract}

\section{Introduction}
In this paper, we study von Neumann--Morgenstern stability (hereafter, vNM stability) in one-to-one matching markets. We ask whether a solution concept defined through dominance relations among sets of outcomes admits a tractable and economically meaningful characterization in the classical one-to-one matching model. Our analysis answers this question affirmatively by showing that, in this setting, the vNM stable set approach admits a sharp structural description based on standard tools from matching theory. More generally, we show that a cooperative solution concept defined through dominance relations can be computed by solving a classical core problem on an explicitly constructed reduced preference profile.

We consider a standard one-to-one matching market with strict preferences on both sides of the market. Such environments provide a natural benchmark for our analysis: they combine a simple coalition structure, in which deviations arise through firms and workers rematching, with one of the best understood notions of stability in matching theory.

Unlike the core or pairwise stability, which identify individual matchings, a vNM stable set is a collection of matchings satisfying two complementary requirements. Internal stability requires that no matching in the set dominates another matching in the same set, while external stability requires that every matching outside the set be dominated by some matching in the set. Thus, a vNM stable set can be viewed as a solution set consisting of mutually non-comparable admissible outcomes, whereas every excluded outcome is dominated by an outcome that belongs to the set.

Since their introduction by \cite{von1947theory}, stable sets have been regarded as one of the fundamental solution concepts in cooperative game theory. Their appeal stems from defining stability through dominance relations among feasible outcomes rather than solely through the absence of blocking coalitions. At the same time, however, vNM stable sets have proved notoriously difficult to analyze. As emphasized by \citet{aumann1987game}, the recursive interaction between internal and external stability makes them considerably harder to characterize than most classical cooperative solution concepts. Consequently, questions concerning existence, uniqueness, structure, and computation remain highly dependent on the particular environment under consideration. Understanding the structure of vNM stable sets in matching markets therefore becomes a natural and important question.

Our main contribution is a complete characterization of the vNM stable set in one-to-one matching markets. We prove that the vNM stable set is unique and coincides with the core of an explicitly constructed reduced preference profile. Thus, a cooperative solution concept defined through dominance relations among sets of matchings admits a representation in terms of the familiar core of a reduced matching market.

The reduced preference profile is constructed through two complementary reduction procedures. The first exploits the lattice structure of the core. Starting from the firm-optimal and worker-optimal stable matchings, we recursively apply the reduction procedure of \cite{irving1986complexity}, recording the pairs that lose mutual acceptability during the construction of the lattice. The second reduction is based on the Efficiency-Adjusted Deferred Acceptance mechanism of \cite{kesten2010school}, which identifies the interrupter pairs responsible for preventing efficiency-improving reallocations. Combining these two reductions yields the reduced preference profile underlying our characterization.

A central technical contribution of the paper is a generalization of the classical Decomposition Lemma. While the original result describes the relationship between core matchings, we show that an analogous decomposition extends to any pair of matchings belonging to the same internally stable set. In particular, the differences between such matchings can be decomposed into disjoint preference cycles connecting firms and workers. Besides its role in proving the main theorem, this decomposition provides a new structural characterization of internally stable sets that may be of independent interest.

Beyond the characterization itself, our results have implications for the dynamics of stability. Using the structural properties of the reduced preference profile, we study the classical path-to-stability problem and show that every matching can be connected to a core matching through a remarkably short sequence of dominance steps. In particular, every matching outside the core reaches a core matching in at most three steps. This provides a constructive description of how stability can be attained from arbitrary initial outcomes and illustrates the dynamic content of our structural characterization.

Our results therefore establish a direct connection between dominance-based cooperative solution concepts and classical matching theory. Besides yielding existence, uniqueness, and a constructive representation of the vNM stable set, they also provide new insight into the dynamics of matching markets by establishing a short path from arbitrary matchings to core outcomes.

\subsection*{Related literature}

These contributions place our paper at the intersection of several strands of the matching literature. We begin with papers that study the existence, structure, and computation of the classical vNM (vNM) stable set in one-to-one matching markets, which are the closest to our contribution.

One of the earliest contributions in this direction is \cite{ehlers2007neumann}, who analyzes the properties of vNM stable sets in one-to-one matching markets and shows that they may differ substantially from the core. In particular, vNM stable sets may contain matchings outside the core and need not admit a simple characterization in terms of standard stability concepts.

Building on this line of research, \cite{faenza2022legal} study the legal set solution concept introduced by \cite{morrill2016making} and prove that the legal set coincides with the core of a suitably reduced instance. In one-to-one matching markets, they further show that, under an appropriate dominance relation, the legal set coincides with the vNM stable set. Outside this setting, however, the two solution concepts are no longer equivalent. Both papers establish representation results through reduced matching instances. However, the underlying solution concepts, the reduction procedures, and the mathematical arguments are fundamentally different. While their reduction is driven by legality, ours is derived directly from the dominance structure underlying vNM stability through two complementary reduction procedures.

More recently, \cite{faenza2025neumann} introduce the notion of internally closed sets of matchings and investigate their relationship with vNM stability in both marriage and roommate markets. They provide new structural characterizations based on generalized rotations and prove polynomial-time algorithms for internally closed sets in the marriage model together with hardness results for the roommate model. Their work considerably broadens the structural understanding of vNM-type solution concepts. In contrast, our paper focuses on the classical vNM stable set itself and establishes a representation theorem showing that the unique vNM stable set coincides with the core of a suitably reduced preference profile.

A second strand of the literature develops algorithmic approaches to vNM stability. \cite{wako2010polynomial} proves existence and uniqueness of the vNM stable set in the marriage model and develops a polynomial-time algorithm that iteratively constructs a modified preference profile whose core coincides with the vNM stable set. In contrast, our characterization is not algorithmic in nature. Rather than repeatedly modifying the preference profile according to successive sets of undominated matchings, we identify two explicit structural reductions that directly yield a reduced preference profile whose core is precisely the unique vNM stable set.

Our analysis is also closely related to the literature on rotations, lattice structures, and preference cycles. Classical contributions by \cite{irving1986complexity} and \cite{gusfield1987three}, together with the exposition of \cite{roth1992two}, use cyclic structures to characterize the lattice of stable matchings and to develop efficient algorithms. These ideas were later extended to many-to-one and many-to-many markets with responsive and substitutable preferences \citep{cheng2008unified,bansal2007polynomial,bonifacio2022cycles}. Preference cycles have also proved useful in the study of strongly stable random matchings \citep{neme2019characterization,neme2021many} and, more recently, in the analysis of vacancy chains and re-stabilization processes \citep{bonifacio2026counting}. Our Decomposition Lemma extends this line of research by showing that the classical decomposition of differences between stable matchings naturally extends to arbitrary internally stable sets. The resulting IS-cycles constitute one of the key structural ingredients of our characterization of the unique vNM stable set.

A second ingredient of our reduction is closely related to the literature on efficiency-adjusted matching mechanisms. The Efficiency-Adjusted Deferred Acceptance (EADA) mechanism introduced by \cite{kesten2010school} eliminates interrupter pairs while preserving Pareto improvements over stable outcomes. More generally, efficiency-motivated reductions have also been studied through alternative notions of stability such as essential stability \citep{troyan2020essentially} and weak stability \citep{tang2021weak}. Unlike these papers, we do not weaken the stability concept. Instead, we exploit the structural properties of EADA mechanism to identify one of the two reductions leading to the reduced preference profile while preserving the classical notion of vNM stability.

Our analysis is also related to the literature on convergence and paths to stability in matching markets. Beginning with \cite{roth1990random}, several papers study whether arbitrary initial matchings can be transformed into stable matchings through sequences of blocking improvements. Related convergence results have subsequently been established for one-sided and many-to-one matching environments, among others by \cite{chung2000}, \cite{klaus2005}, \cite{kojima2008}, \cite{eriksson2008}, and \cite{diamantoudi2004}. More generally, structural characterizations have also been used to study convergence to stability in coalition formation games \citep{bonifacio2024absorbing}. Our structural characterization of the vNM stable set yields a bounded path-to-stability result: every individually rational matching reaches a core matching through at most three coalitional dominance steps.

The rest of the paper is organized as follows. Section~\ref{seccion Preliminares} introduces the matching model, the dominance relation, the notion of vNM stability, and the basic concepts used throughout the paper. Section~\ref{seccion Construccion del VNM stable set} develops the main theoretical results. In particular, it establishes the Decomposition Lemma for internally stable sets, introduces the reduced preference profiles that underlie our analysis, constructs the candidate vNM stable set through the core of the reduced profile, and provides the characterization of the unique vNM stable set. Section~\ref{seccion path to core} studies the path-to-stability problem. Using the structural results established for the vNM stable set, we show how an arbitrary matching can be connected to a core matching through a sequence of dominance steps. Section~\ref{conclusiones} concludes and discusses directions for future research. Appendix~\ref{apendice pruebas} contains all omitted proofs, while Appendix~\ref{apendice descripcion EADAM} presents the adaptation of the EADA mechanism to our environment and summarizes the properties used in the paper.

\section{Model and preliminaries}\label{seccion Preliminares}
 We consider one-to-one matching markets where there are two disjoint sets of agents: the set of \textit{firms} $F$ and the set of \textit{workers} $W$. Each firm \(f\in F\) has a strict preference ordering \(P_f\) over
\(W\cup\{\emptyset\}\), and each worker \(w\in W\) has a strict preference
ordering \(P_w\) over \(F\cup\{\emptyset\}\).
A preference profile is denoted by $P=(P_f)_{f\in F}\times(P_w)_{w\in W}.$
  Since the sets $F$ and $W$ are kept fixed throughout the paper, we often identify the market $P$ with the preference profile $P$.  Given an agent $a \in F\cup W$, an agent
 $b$ in the opposite side of the market is \textbf{acceptable for  $\boldsymbol{a}$ under $\boldsymbol{P}$} if $bP_a\emptyset$. A pair $(w, f) \in W \times F$ is an \textbf{acceptable pair under $\boldsymbol{P}$ } if $f$ is acceptable for $w$ under $P$ and
 $w$ is acceptable for $f$ under $P$. The preference relation $P_a$ is represented by the ordered list of its acceptable agents (from most to least preferred).\footnote {For instance, $P_f:w_1,w_4,w_2$ indicates that $w_1P_fw_4P_fw_2P_f\emptyset$ and $P_w:f_1,f_2,f_3$ indicates that $f_1P_wf_2P_wf_3P_w\emptyset$.}

A \textbf{matching} $\mu$ is a function from $F\cup W$ into
$F\cup W\cup\{\emptyset\}$ such that:
(i) $\mu(w)\in F\cup\{\emptyset\}$ for every $w\in W$;
(ii) $\mu(f)\in W\cup\{\emptyset\}$ for every $f\in F$; and
(iii) $\mu(f)=w$ if and only if $\mu(w)=f$ for every $(f,w)\in F\times W$. The \textbf{set of matching} is denoted by $ \boldsymbol{\mathcal{M}}.$

Agent $a\in F\cup W$ is \textbf{matched} if $\mu(a) \neq \emptyset$, otherwise she is \textbf{unmatched}. For the following definitions, fix a preference profile $P$.
A matching $\mu$ is \textbf{blocked by agent $\boldsymbol{a}$} if $\emptyset P_a \mu(a)$. A matching is  \textbf{individually rational} if it is not blocked by any individual agent. The \textbf{set of individually rational matching} is denoted by $\boldsymbol{\mathcal{IR}}.$

Given a preference profile $P$ and two  matchings $\mu$ and  $\mu'$,  we write  $\boldsymbol{\mu \succ_W \mu'}$ whenever $\mu(w) P_w \mu'(w)$ for each $w\in W$, and we write  $\boldsymbol{\mu \succeq_W \mu'}$ whenever $\mu(w) R_w \mu'(w)$ for each $w\in W$, and there is $w'\in W$ such that $\mu(w') P_{w'}\mu'(w')$. Similarly, we write $\boldsymbol{\mu \succ_F\mu'}$ and $\boldsymbol{\mu \succeq_F\mu'}$. Note that $\succeq_W$ and $\succeq_F$, which represent the common preferences of the firms and workers, respectively, are partial orders over the set of matchings.


\begin{definition}[Dominance relation]
Let $\mu, \mu' \in \mathcal{M}$ and $S \subseteq F \cup W$. We say that \textbf{\boldmath{$\mu$} dominates \boldmath{$\mu'$} via \boldmath{$S$}}, denoted by $\mu \succ_{S} \mu'$, if:
\begin{enumerate}
    \item $\mu(S) = S$;
    \item $\mu(a) P_{a} \mu'(a)$ for all $a \in S$.
\end{enumerate}
In the case where there is a coalition \(S\) such that \(\mu\) dominates \(\mu'\) via \(S\), 
we refer to \(S\) as a blocking coalition.
\end{definition}
The dominance relation introduced above provides a natural way to define the 
core of the market. In particular, the core consists precisely of those 
matchings that cannot be dominated through a deviation by any coalition.
\begin{definition}\label{definicion core}

The \textbf{core} of a market $P$ is the set of undominated matchings, 
denoted by $C(P)$. That is,
    
\begin{center}
$C(P) = \{\mu \in \mathcal{M} \mid \text{there is no coalition } S \subseteq F \cup W 
\text{ and no matching } \mu' \in \mathcal{M} \text{ such that } \mu' \succ_S \mu \}$.
\end{center}

\end{definition}

It is well known from \cite{gale1962college} that the core of every one-to-one matching market is nonempty. Therefore, $C(P)\neq\emptyset$ throughout the paper.

A pair of agents, one from each side of the market $P$, is said to be a (corewise) \textbf{stable pair of $\boldsymbol{P}$} if there is a core matching that matches them.

After introducing the core, it is useful to recall a key structural property of this set.  
\cite{knuth1976marriages} shows that the collection of core matchings forms a dual lattice under the partial orders $\succeq_F$ and $\succeq_W$.  
In particular, the lattice admits two extremal elements: the \textbf{firm-optimal core matching} with respect to $\succeq_F$, denoted by $\boldsymbol{\mu_F}$, and the \textbf{worker-optimal core matching} with respect to $\succeq_W$, denoted by $\boldsymbol{\mu_W}$.  
Moreover, these extremal matchings are also worst for the opposite side: $\mu_F$ is the \textbf{worker-pessimal core matching} under $\succeq_W$, while $\mu_W$ is the \textbf{firm-pessimal core matching} under $\succeq_F$ (see \citealp{roth1992two} for further details).
This lattice structure highlights the rich geometry of the core and clarifies how core outcomes can be partially ordered according to the preferences of each side of the market.  
Building on this dominance perspective, we now turn to a broader solution concept that evaluates sets of matchings rather than individual ones.

The dominance relation introduced above naturally leads to the two requirements that define the solution concept studied in this paper: the vNM stability notion. 
Accordingly, we consider subsets of individually rational matchings $V$, viewed as candidates for vNM stable sets.

The first requirement is \textbf{internal stability}: no matching in $V$ should
dominate another matching in the same set. Formally, $V$ satisfies internal
stability if there exist no matchings $\mu,\mu' \in V$ and no coalition
$S \subseteq F \cup W$ such that $\mu \succ_S \mu'$.
Internal stability rules out internal conflicts within the set: no matching in
$V$ can be justified as an improvement over another matching in $V$ from the
perspective of some coalition. We denote by $\boldsymbol{V_{IS}(P)}$ the collection of all sets that satisfy internal stability.

The second requirement is \textbf{external stability}: every matching outside
the set must be dominated by at least one matching in $V$. Formally, $V$
satisfies external stability if, for every $\mu' \in \mathcal{M}\setminus V$,
there exist a coalition $S \subseteq F \cup W$ and a matching $\mu \in V$ such that
$\mu \succ_S \mu'$. Hence, external stability ensures that whenever a matching
lies outside the set, the set already contains an alternative that some coalition
strictly prefers to it. We denote by $\boldsymbol{V_{ES}(P)}$ the collection of all sets that satisfy external stability.

\begin{definition}\label{def:vNM}
A set $V \subseteq \mathcal{IR}$ is a \textbf{von Neumann--Morgenstern stable
set (\textbf{vNM stable set})} if it satisfies both internal and external stability.
\end{definition}

Let $\mathrm{vNM}(P)$ denote the collection of all vNM stable sets of market $P$. A fundamental observation established by \cite{ehlers2007neumann} is that every vNM stable set contains the core. Formally,

\begin{remark}\label{remark de core en von neuman}
    $C(P) \subseteq V$ for every stable set $V \in \mathrm{vNM}(P)$.
\end{remark}

By the definition of internal stability, it is immediate that the core of a market satisfies this property. 
This is not the case for external stability; however, there are markets in which the core satisfies it as well. 
The following example illustrates that the core is sometimes a vNM stable set. 

\begin{example}
    
Consider the matching market $(F,W,P)$, with $F=\{f_1,f_2,f_3\}$ and $W=\{w_1,w_2,w_3\}$, and suppose that preferences are given by

\begin{center}
\begin{tabular}{l l @{\hspace{1cm}} l l}
$P_{f_{1}}:$ & $w_{2}$, $w_{1}$, $w_{3}$ & 
$P_{w_{1}}:$ & $f_{1}$, $f_{3}$, $f_{2}$ \\
$P_{f_{2}}:$ & $w_{1}$, $w_{3}$, $w_{2}$ & 
$P_{w_{2}}:$ & $f_{3}$, $f_{1}$, $f_{2}$ \\
$P_{f_{3}}:$ & $w_{1}$, $w_{2}$, $w_{3}$ & 
$P_{w_{3}}:$ & $f_{1}$, $f_{3}$, $f_{2}$ \\
\end{tabular}
\end{center}
Since all pairs are mutually acceptable, every matching is individually rational. The set of all matchings is
\[
\mu_F=\begin{pmatrix}
f_1 & f_2 & f_3\\
w_2 & w_3 & w_1
\end{pmatrix},\qquad
\mu_W=\begin{pmatrix}
f_1 & f_2 & f_3\\
w_1 & w_3 & w_2
\end{pmatrix},\qquad
\mu_1=\begin{pmatrix}
f_1 & f_2 & f_3\\
w_1 & w_2 & w_3
\end{pmatrix},
\]
\[
\mu_2=\begin{pmatrix}
f_1 & f_2 & f_3\\
w_2 & w_1 & w_3
\end{pmatrix},\qquad
\mu_3=\begin{pmatrix}
f_1 & f_2 & f_3\\
w_3 & w_1 & w_2
\end{pmatrix},\qquad
\mu_4=\begin{pmatrix}
f_1 & f_2 & f_3\\
w_3 & w_2 & w_1
\end{pmatrix}.
\]
Then, the core is given by
$C(P)=\{\mu_F,\mu_W\}$%
\footnote{The firm-optimal and worker-optimal core matchings are obtained by the Deferred Acceptance algorithm of \cite{gale1962college}, with firms and workers proposing, respectively. The complete set of core matchings is then obtained using the algorithm of \cite{gusfield1989stable}.}

Internal stability follows immediately, since neither $\mu_F$ nor $\mu_W$ dominates the other. Moreover, the remaining matchings are dominated by elements of the core:
\[
\mu_F \succ_{\{f_1,w_2\}} \mu_1,\qquad
\mu_W \succ_{\{f_3,w_2\}} \mu_2,\qquad
\mu_W \succ_{\{f_1,w_1\}} \mu_3,\qquad
\mu_W \succ_{\{f_1,w_1\}} \mu_4.
\]
Therefore, external stability also holds. We conclude that $C(P) \in vNM(P)$.

\end{example}

Before proceeding with the general analysis, it is useful to illustrate that the core of a matching market does not necessarily constitute a vNM stable set. Although the core is internally stable by definition---since no core matching can dominate another core matching---it may fail to satisfy external stability. More precisely, there may exist matchings outside the core that are not dominated by any matching in the core. In such cases, the core cannot itself be a vNM stable set. The following example, adapted from \cite{wako2008note}, illustrates this possibility.
\begin{example}[\citealp{wako2008note}]\label{ex:wako}

Let $F=\{f_1,f_2,f_3,f_4\}$ be the set of firms and $W=\{w_1,w_2,w_3,w_4\}$ the set of workers. The preference profile $P$ is given by

\begin{center}
\begin{tabular}{l l @{\hspace{1cm}} l l}
$P_{f_1}:$ & $w_2, w_4, w_1, w_3$ &
$P_{w_1}:$ & $f_1, f_2, f_3, f_4$ \\
$P_{f_2}:$ & $w_3, w_1, w_2, w_4$ &
$P_{w_2}:$ & $f_2, f_4, f_3, f_1$ \\
$P_{f_3}:$ & $w_3, w_2, w_1, w_4$ &
$P_{w_3}:$ & $f_4, f_2, f_3, f_1$ \\
$P_{f_4}:$ & $w_4, w_2, w_3, w_1$ &
$P_{w_4}:$ & $f_3, f_1, f_4, f_2$ \\
\end{tabular}
\end{center}

The core of this market is

$$
C(P)=\{\mu_F,\mu_1,\mu_W\},
$$

where

$$
\mu_F=
\begin{pmatrix}
f_1 & f_2 & f_3 & f_4\\
w_4 & w_3 & w_1 & w_2
\end{pmatrix},
\qquad
\mu_1=
\begin{pmatrix}
f_1 & f_2 & f_3 & f_4\\
w_1 & w_3 & w_4 & w_2
\end{pmatrix},
\qquad
\mu_W=
\begin{pmatrix}
f_1 & f_2 & f_3 & f_4\\
w_1 & w_2 & w_4 & w_3
\end{pmatrix}.
$$

Now consider the matching

$$
\widetilde{\mu}=
\begin{pmatrix}
f_1 & f_2 & f_3 & f_4\\
w_4 & w_2 & w_1 & w_3
\end{pmatrix}.
$$

Let $\mu$ be the matching obtained from $\widetilde{\mu}$ by matching $f_2$ with $w_1$, leaving their former partners unmatched, and keeping all other pairs unchanged. Thus,
$$
\mu=
\begin{pmatrix}
f_1 & f_2 & f_3 & f_4\\
w_4 & w_1 & \emptyset & w_3
\end{pmatrix}.
$$
Then, $\mu(f_2)=w_1 \,P_{f_2}\, w_2=\widetilde{\mu}(f_2),$ and $\mu(w_1)=f_2 \,P_{w_1}\, f_3=\widetilde{\mu}(w_1).$
Therefore, $\mu \succ_{\{f_2,w_1\}} \widetilde{\mu},$
so $\widetilde{\mu}\notin C(P)$.
Moreover, $(f_2,w_1)$ is an unstable pair. Hence, no core matching can match $f_2$ with $w_1$, and therefore no core matching can dominate $\widetilde{\mu}$ via the coalition $\{f_2,w_1 \}$. In fact, no core matching dominates $\widetilde{\mu}$. It follows that $C(P)$ is not externally stable, and thus $C(P)$ is not a vNM stable set.

\end{example}

The previous example illustrates that the core need not satisfy external stability and therefore may fail to constitute a vNM stable set. The remainder of the paper is devoted to characterizing the unique enlargement of the core that restores external stability while preserving internal stability.

\section{Constructing the von Neumann--Morgenstern stable set}\label{seccion Construccion del VNM stable set}

The purpose of this section is to construct the vNM stable set. We begin by introducing a generalized notion of preference cycles, building on the classical construction of \cite{irving1986complexity}. These cycles allow us to define suitable reduced preference profiles whose cores will be used in the characterization established in Subsection~\ref{seccion caracterizacion de VNM}.

We first establish a generalization of the classical Decomposition Lemma from matching theory \citep[see][for its formulation in the context of core matchings]{roth1992two}.\footnote{Related generalizations of the Decomposition Lemma for many-to-one matching environments are provided by \cite{ehlers2020legal}.} While the original result describes the structure of differences between core matchings, our version applies to any pair of matchings belonging to the same internally stable set. Although this result is fundamental for the construction and characterization developed in this paper, we believe that the generalization is of independent interest, as it provides a structural description of internally stable sets that extends beyond the analysis of the core.

The intuition is simple. When two matchings belong to the same internally stable set, differences between them must be balanced across the two sides of the market. In particular, if some agents are strictly better off under one matching, then agents on the opposite side must be strictly worse off, since otherwise one matching would dominate the other. To formalize this result, let
\[
F(\mu',\mu)=\{f\in F:\mu'(f)P_f\mu(f)\}
\quad\text{and}\quad
W(\mu',\mu)=\{w\in W:\mu'(w)P_w\mu(w)\}
\]
denote the sets of firms and workers, respectively, that strictly prefer their match under \(\mu\) to their match under \(\mu'\). The lemma below formalizes this observation for internally stable sets.

\begin{lemma}[Decomposition Lemma]\label{lema descomposicion para estables}
    Let $P$ be a one-to-one matching model with strict preferences, and let $V_{IS} \in V_{IS}(P)$. If $\mu, \mu' \in V_{IS}$, then $\mu$ and $\mu'$ map $F(\mu',\mu)$ onto $W(\mu,\mu')$ and $F(\mu,\mu')$ onto $W(\mu',\mu)$.
\end{lemma}
\begin{proof}
The proof is relegated to Appendix \ref{apendice pruebas}.
\end{proof}

The Decomposition Lemma reveals that differences between two matchings in an internally stable set are organized through a precise correspondence between agents on the two sides of the market. 
This structure suggests that transitions from one matching to another can be understood as sequences of mutually linked deviations. 
Motivated by this observation, we introduce a notion of cycles that captures how agents are reallocated when moving between matchings in $V_{IS}(P)$.

\begin{definition}
Let $V_{IS} \in V_{IS}(P)$ and let $\mu', \mu \in V_{IS}$. We say that $O^{\mu',\mu}$ is an \textbf{IS-cycle between $\boldsymbol{\mu'$ and $\mu}$} if there is a set of firms $\{f_1,\ldots,f_k\}\subseteq F(\mu',\mu)$ and a set of workers $\{w_0,\ldots,w_{k-1}\}$ such that $\mu(f_i)=w_{i-1}$ for all $i=1,\ldots,k$, and $\mu'(f_j) = w_j$ for $j=1,\ldots,k-1$, with $\mu'(f_k) = w_0$.\footnote{The phrase ``between $\mu'$ and $\mu$'' is omitted when the context is clear, as the notation $O^{\mu',\mu}$ already indicates the pair of matchings involved.}
\end{definition}
The IS-cycle $O^{\mu',\mu}$ is represented by the ordered sequence of workers and firms $$O^{\mu',\mu} = ( w_0, f_1, w_1, f_2, \ldots, w_{k-1}, f_k).$$
An IS-cycle describes a closed chain of reallocations connecting $\mu$ and $\mu'$. 
Starting from a worker who changes partner when moving from $\mu$ to $\mu'$, each firm in the cycle is matched under $\mu$ with the previous worker in the sequence and under $\mu'$ with the next one, until the sequence closes. 
Moreover, the workers involved in the cycle necessarily belong to $W(\mu,\mu')$, so the cycle isolates a reallocation in which firms that gain under $\mu'$ are paired with workers who lose relative to $\mu$. In other words,  $O^{\mu',\mu} \cap F \subseteq F(\mu', \mu)$ and $O^{\mu',\mu} \cap W \subseteq W(\mu, \mu')$.

The next result shows that such cycles always arise whenever the two matchings differ on the firm side. 
In particular, any nonempty set $F(\mu',\mu)$ generates at least one IS-cycle, which identifies a fundamental building block of the differences between the two matchings.

\begin{proposition}\label{existencia de un IS-cycle}
Let $V_{IS} \in V_{IS}(P)$ and let $\mu', \mu \in V_{IS}$. Whenever $F(\mu', \mu) \neq \emptyset$, there is an IS-cycle $O^{\mu',\mu}.$ 
\end{proposition}
\begin{proof}
 The proof is relegated to Appendix \ref{apendice pruebas}.
\end{proof}

We now strengthen the previous existence result by showing that the differences between $\mu$ and $\mu'$ can be fully decomposed into disjoint IS-cycles. 
Rather than appearing in an arbitrary way, firms that prefer $\mu'$ to $\mu$ are organized into independent cyclic structures, each capturing a self-contained reallocation between the two matchings. 
Consequently, the collection of IS-cycles provides a complete description of the set $F(\mu',\mu)$ (and, symmetrically, of $W(\mu,\mu')$).

Returning to Example \ref{ex:wako}, we observe that there are two IS-cycles between  \(\mu_F\) and \(\mu_W\):
\[
O_1^{\mu_F,\mu_W}=\{w_3,f_2,w_2,f_4\}
\]
and
\[
O_2^{\mu_F,\mu_W}=\{w_4,f_1,w_1,f_3\}.
\]
The next theorem formalizes this decomposition.

\begin{theorem}\label{Teorema descomposicion ciclos disjuntos}
 Let $V_{IS}\in V_{IS}(P)$ and let $\mu',\mu\in V_{IS}$.
Then, there is a collection of pairwise disjoint IS-cycles
$O_1^{\mu',\mu},\ldots,O_k^{\mu',\mu}$ such that
$$
\left(
\bigcup_{i=1}^{k}O_i^{\mu',\mu}
\right)\cap F
=
F(\mu',\mu).
$$
\end{theorem}
\begin{proof}
The proof is relegated to Appendix \ref{apendice pruebas}.
\end{proof}

As a direct implication of the Decomposition Lemma together with 
Theorem~\ref{Teorema descomposicion ciclos disjuntos}, we obtain the following corollary, 
which provides the analogous decomposition on the worker side.

\begin{corollary}
Let $V_{IS}\in V_{IS}(P)$ and let $\mu',\mu\in V_{IS}$. There is a
collection of pairwise disjoint IS-cycles
$O_1^{\mu',\mu},\ldots,O_k^{\mu',\mu}$ such that
$$
\left(
\bigcup_{i=1}^{k}O_i^{\mu',\mu}
\right)\cap W
=
W(\mu,\mu').
$$
\end{corollary}

The next results focus on non-core matchings with specific structural properties. 
The following proposition identifies situations in which a matching $\mu$ cannot belong to an internally stable set together with one of the two extremal core matchings.
Intuitively, if $\mu$ is strictly worse for one side relative to the corresponding extremal matching and no allocation offsets this loss, then internal stability cannot be maintained.
\begin{proposition}\label{proposition mu no estable fuera de VI_s}
Let  $\mu \in \mathcal{IR}$. Assume that $F(\mu_W, \mu) \neq \emptyset$, and that no matching $\mu_1 \in \mathcal{IR}$ satisfies $\mu_1 \succeq_W \mu_W$. Then, there is no set $V_{IS} \in V_{IS}(P)$ such that $\mu$ and $\mu_W$ both belong to $V_{IS}$.
\end{proposition}
\begin{proof}
The proof is relegated to Appendix \ref{apendice pruebas}. 
\end{proof}

Note that, by symmetry of the model, if $W(\mu_F,\mu) \neq \emptyset$ and no matching 
$\mu_2 \in \mathcal{IR}$ satisfies $\mu_2 \succeq_F \mu_F$, then there is no set 
$V_{IS} \in V_{IS}(P)$ such that $\mu$ and $\mu_F$ both belong to $V_{IS}$.

The next proposition identifies a local condition on a matching that prevents it from being part of an internally stable set together with the core allocations. In particular, it shows that certain deviations at the level of a single pair are sufficient to violate internal stability.
The result highlights that incompatibility with internal stability may arise from the behavior of a single matched pair.

\begin{proposition}\label{Proposition Resultado 2}
Let $\mu \in \mathcal{IR}$ and let $(f,w)$ be a pair such that $\mu(f) = w$. 
If $(f,w)$ is not a stable pair of $P$ and 
$\mu_F(f) P_f \mu(f)  P_f  \mu_W(f)$, 
then  there is no set $V_{IS} \in V_{IS}(P)$ such that $\{\mu\} \cup C(P) \subseteq V_{IS}$.
\end{proposition}
\begin{proof}
The proof is relegated to Appendix \ref{apendice pruebas}.
\end{proof}

\subsection{Subprofiles and their structural properties}\label{subseccion de construcion y resultados de subperfiles}

We begin by introducing the formal definition of subprofiles. As anticipated earlier, this notion plays a central role in the construction of the vNM stable set.

\begin{definition}
Let $f \in F$ and let $P_f$ be its preference list. We say that \textbf{\boldmath{$P'_f$} is a subpreference} of $P_f$ if the following conditions hold:
\begin{enumerate}
    \item If $w P'_f \emptyset$, then $w P_f \emptyset$.
    \item If $w P'_f w' P'_f \emptyset$, then $w P_f w'$.
\end{enumerate}
\end{definition}

Analogously, we define a subpreference for each worker $w \in W$, denoted by $P'_w$. Given an agent $a \in F\cup W$ and two preference lists $P_a$ and $P'_a$, we write $\boldsymbol{P'_a \sqsubseteq P_a}$ whenever $P'_a$ is a subpreference of $P_a$.
We write $\boldsymbol{P' \sqsubseteq P}$ whenever $P'_a \sqsubseteq P_a$ for every agent $a \in F\cup W$. In this case, we say that $P'$ is a \textbf{subprofile} of $P$.
We restrict attention to subprofiles in which acceptability remains mutual. That is, for every $(f,w)\in F\times W$,
\[
wP'_f\emptyset
\quad\Longleftrightarrow\quad
fP'_w\emptyset.
\]

Given $f\in F$, if any $w \in W$ is such that $w P_f \emptyset$ and $\emptyset P'_f w$, we say that $w$ is eliminated from the preference list $P_f$. Analogously, we define when $f$ is eliminated from the preference list $P_w$.

A subprofile $P'$ of $P$ \textbf{eliminates a pair $(f,w)$} whenever $w$ is eliminated from the preference list $P_f$ and $f$ is eliminated from the preference list $P_w$.
We denote by $\boldsymbol{\mathcal{E}^{P\to P'}}$ the set of all pairs eliminated from $P$ by $P'$.

The next three lemmas examine how the core behaves under restrictions of the preference profile. These results describe how stability properties are preserved when moving from the original profile to specific subprofiles, and will serve as key tools in the analysis that follows.

The first lemma shows that removing no pair that is stable under the original profile preserves every matching in the original core. Hence, restricting preferences in this way can only enlarge the core.

\begin{lemma}\label{core es subconjunto de core de subperfil}
  
    Let $\widetilde{P}$ be a preference profile such that $\widetilde{P} \sqsubseteq P$. If the set $\mathcal{E}^{P\to \widetilde{P}}$ contains no stable pair of $P$ , then $C(P) \subseteq C(\widetilde{P})$.
\end{lemma}
\begin{proof}
  The proof is relegated to Appendix \ref{apendice pruebas}.
\end{proof}

The second lemma shows that the core of the reduced profile is internally stable when dominance is evaluated under the original profile. Thus, it provides a natural candidate internally stable set for the original market.

\begin{lemma}\label{todo core cumple estabilidad interna}
    Let $\widetilde{P}$ be a preference profile such that $\widetilde{P} \sqsubseteq P$. If the set $\mathcal{E}^{P\to \widetilde{P}}$ contains no stable pair  of $P$, then $C(\widetilde{P}) \in V_{IS}(P)$.
\end{lemma}
\begin{proof}
   The proof is relegated to Appendix \ref{apendice pruebas}.
\end{proof}

When $\widetilde P \sqsubseteq P$, Lemma~\ref{core es subconjunto de core de subperfil} shows that removing no stable pair of $P$ can only enlarge the core. The next lemma extends this observation to successive restrictions of the preference profile.
It shows that if, at each stage, the restriction removes no pair that is stable under the profile prevailing at that stage, then the corresponding cores form a nested sequence.

\begin{lemma}\label{core de subperfil subconjunto de core de otro subperfil}
    Let $P_1$ and $P_2$ be preference profiles such that  $P_2 \sqsubseteq P_1 \sqsubseteq P$. If the set $\mathcal{E}^{P\to P_1}$ contains no stable pair of $P$  and the set $\mathcal{E}^{P_1\to P_2}$ contains no stable pair of $P_1$ , then  $C(P) \subseteq C(P_1) \subseteq C(P_2)$.
\end{lemma}
\begin{proof}
   The proof is relegated to Appendix \ref{apendice pruebas}.
\end{proof}

Taken together, the three lemmas describe how the core evolves under restrictions of the preference profile. 
When pairs are removed without eliminating any stable pair, the core is preserved and may only expand. 
Moreover, the core of each subprofile forms an internally stable set relative to the original profile, and successive restrictions generate a nested sequence of cores. 
These properties provide a consistent structural framework for analyzing subprofiles and will be instrumental in the construction that follows.

\subsection{Subprofiles induced by the Irving--Leather and EADA reductions}\label{subseccion subperfiles}

We first construct a subprofile induced by the classical Irving--Leather
reduction procedure \citep{irving1986complexity}, presented in
\cite{roth1992two}. A detailed description of the recursive reduction is
provided in Appendix~\ref{appendix:IL}. Here we recall only the aspects that
are relevant for our construction.

Starting from the original preference profile $P$, the procedure recursively
traverses the lattice of core matchings, beginning with the firm-optimal and
worker-optimal core matchings, $\mu_F$ and $\mu_W$. At each recursive step,
the preference profile is restricted to the lattice interval determined by
the pair of core matchings under consideration. Mutual acceptability is then
restored by eliminating every pair $(f,w)$ such that one agent still finds
the other acceptable whereas the reverse is no longer true.

The only information that we retain from this reduction procedure is the
collection of pairs eliminated during this mutual acceptability restoration
step. We denote this collection by $\mathcal E^{IL}$.

\begin{definition}\label{defincion subperfil con ciclos}
The subprofile $SP^{\mu_F,\mu_W}$ is obtained from $P$ by eliminating every pair in $\mathcal E^{IL}$.
\end{definition}

By construction,
\[
\mathcal E^{P\to SP^{\mu_F,\mu_W}}
=
\mathcal E^{IL}.
\]


Hence, the general notation introduced in the previous subsection coincides,
in this particular case, with the collection of pairs eliminated by the
Irving--Leather reduction procedure.
The following remark characterizes the pairs belonging to
$\mathcal E^{IL}$ (equivalently,
$\mathcal E^{P\to SP^{\mu_F,\mu_W}}$).

\begin{remark}\label{remark el conjunto de pares eliminados para muF y muW no contiene pares estables}
The set $\mathcal{E}^{IL}$ contains no stable pairs under $P$.
More precisely, for each pair $(f,w)\in
\mathcal{E}^{IL}$, there is a core matching
$\mu\in C(P)$ such that either
\[
\mu(f) P_f w P_f \mu_W(f)
\qquad \text{and} \qquad
\mu(w) P_w f,
\]
or
\[
\mu_W(w) P_w f P_w \mu(w)
\qquad \text{and} \qquad
\mu_W(f) P_f w.
\]
\end{remark}

The next step is to relate the core of the original profile with the core of the restricted subprofile.

Since $SP^{\mu_F,\mu_W}\sqsubseteq P$ and, by
Remark~\ref{remark el conjunto de pares eliminados para muF y muW no contiene pares estables}, the set
$\mathcal E^{IL}$
contains no stable pair of $P$,
Lemma~\ref{core es subconjunto de core de subperfil}
immediately yields
\[
C(P)\subseteq C(SP^{\mu_F,\mu_W}).
\]
Applying
Lemma~\ref{todo core cumple estabilidad interna}
under the same assumptions, we also obtain
\[
C(SP^{\mu_F,\mu_W})\in V_{IS}(P).
\]

Note that $C(P)$ and $C(SP^{\mu_F,\mu_W})$ need not coincide. 
When they differ, the reduced core contains matchings in 
$C(SP^{\mu_F,\mu_W}) \setminus C(P)$ that are not dominated by any core matching 
of the original profile.

\begin{examplecont}{ex:wako}
We now construct the subprofile $SP^{\mu_F,\mu_W}$ as in Definition \ref{defincion subperfil con ciclos}. By Remark~\ref{remark el conjunto de pares eliminados para muF y muW no contiene pares estables}, the only pair removed by the Irving--Leather reduction is $(f_2,w_1).$ 
 Hence, we get \(\mathcal{E}^{IL}\)$=\{(f_2,w_1)\}$.
The resulting subprofile is given by
\begin{center}
\begin{tabular}{l l @{\hspace{1cm}} l l}
$P^{SP^{\mu_F,\mu_W}}_{f_{1}}:$ & $w_{2}$, $w_{4}$, $w_{1}$, $w_3$ &
$P^{SP^{\mu_F,\mu_W}}_{w_{1}}:$ & $f_{1}$, $f_{3}$, $f_4$ \\
$P^{SP^{\mu_F,\mu_W}}_{f_{2}}:$ & $w_{3}$, $w_{2}$, $w_4$ &
$P^{SP^{\mu_F,\mu_W}}_{w_{2}}:$ & $f_{2}$, $f_{4}$, $f_{3}$, $f_1$ \\
$P^{SP^{\mu_F,\mu_W}}_{f_{3}}:$ & $w_{3}$, $w_{2}$, $w_{1}$, $w_4$ &
$P^{SP^{\mu_F,\mu_W}}_{w_{3}}:$ & $f_{4}$, $f_{2}$, $f_{3}$, $f_1$ \\ 
$P^{SP^{\mu_F,\mu_W}}_{f_{4}}:$ & $w_{4}$, $w_{2}$, $w_{3}$, $w_1$ &
$P^{SP^{\mu_F,\mu_W}}_{w_{4}}:$ & $f_{3}$, $f_{1}$, $f_{4}$, $f_2$ \\ 
\end{tabular}
\end{center}
The core of the reduced profile is
\[
C\big(SP^{\mu_F,\mu_W}\big)=\{\mu_F,\mu_1,\widetilde{\mu},\mu_W\},
\]
where
\[
\widetilde{\mu}=
\begin{pmatrix}
f_1 & f_2 & f_3 & f_4\\
w_4 & w_2 & w_1 & w_3
\end{pmatrix}.
\]
In particular, $\widetilde{\mu}\notin C(P)$ but $\widetilde{\mu}\in C\big(SP^{\mu_F,\mu_W}\big)$. Hence,
\[
C(P)\subsetneq C\big(SP^{\mu_F,\mu_W}\big).
\]

Moreover, since $\{f_2,w_1\}$ is the only blocking coalition of $\widetilde{\mu}$ and $(f_2,w_1)$ has been removed in $SP^{\mu_F,\mu_W}$, no matching in $C(P)$ can dominate $\widetilde{\mu}$. Therefore, $\widetilde{\mu}$ is undominated in the reduced profile and $C\big(SP^{\mu_F,\mu_W}\big)$ satisfies internal stability with respect to $P$.
\end{examplecont}

We now construct a second subprofile induced by the
Efficiency-Adjusted Deferred Acceptance (EADA) mechanism of
\cite{kesten2010school}. A detailed description of the mechanism is
provided in Appendix~\ref{apendice descripcion EADAM}. Here we recall
only the aspects that are relevant for our construction.

Starting from a deferred acceptance outcome, the EADA mechanism
iteratively removes interrupter pairs, thereby generating a sequence of
efficiency-improving matchings. Let $\mu_{E_F}$ and $\mu_{E_W}$ denote the outcomes of the EADA mechanism when firms propose and when workers propose, respectively.
We execute the mechanism under both proposing sides and collect every
interrupter pair eliminated during either execution. We denote the
resulting collection by $\mathcal E^{EADA}$.

\begin{definition}\label{definition EADA subprofile}
The subprofile $SP^{\mu_{E_F},\mu_{E_W}}$ is obtained from $P$ by
eliminating every pair in $\mathcal E^{EADA}$.
Equivalently,
\[
\mathcal E^{P\to SP^{\mu_{E_F},\mu_{E_W}}}
=
\mathcal E^{EADA}.
\]
\end{definition}

The following remark characterizes the pairs belonging to
$\mathcal E^{EADA}$ (equivalently,
$\mathcal E^{P\to SP^{\mu_{E_F},\mu_{E_W}}}$).

\begin{remark}\label{remark el conjunto de pares eliminados no contiene pares estables}
The set $\mathcal E^{EADA}$ contains no stable pairs under the original
profile $P$.
More precisely, every pair
$(f,w)\in\mathcal E^{EADA}$ satisfies
\[
w P_f \mu_F(f)
\qquad \text{and} \qquad
\mu_F(w) P_w f,
\]
or
\[
f P_w \mu_W(w)
\qquad \text{and} \qquad
\mu_W(f) P_f w,
\]
depending on which side of the market acts as the proposing side in the
EADA procedure.
\end{remark}

We illustrate the previous construction by returning to
Example~\ref{ex:wako}.
The example shows how the interrupter pairs generated by the two executions
of the EADA mechanism determine the subprofile
$SP^{\mu_{E_F},\mu_{E_W}}$ and enlarge the core.

\begin{examplecont}{ex:wako}
In this instance, $(f_3,w_2)$ is the only interrupter pair in $P$. Hence, by Remark \ref{remark el conjunto de pares eliminados no contiene pares estables}, we get
\[
\mathcal{E}^{EADA}=\{(f_3,w_2)\}.
\]
The resulting subprofile is given by
\begin{center}
\begin{tabular}{l l @{\hspace{1cm}} l l}
$P^{SP^{\mu_{E_F},\mu_{E_W}}}_{f_{1}}:$ & $w_{2}$, $w_{4}$, $w_{1}$, $w_3$ &
$P^{SP^{\mu_{E_F},\mu_{E_W}}}_{w_{1}}:$ & $f_{1}$, $f_{2}$, $f_{3}$, $f_{4}$ \\
$P^{SP^{\mu_{E_F},\mu_{E_W}}}_{f_{2}}:$ & $w_{3}$, $w_{1}$, $w_{2}$, $w_4$ &
$P^{SP^{\mu_{E_F},\mu_{E_W}}}_{w_{2}}:$ & $f_{2}$, $f_{4}$, $f_{1}$ \\
$P^{SP^{\mu_{E_F},\mu_{E_W}}}_{f_{3}}:$ & $w_{3}$, $w_{1}$, $w_{4}$ &
$P^{SP^{\mu_{E_F},\mu_{E_W}}}_{w_{3}}:$ & $f_{4}$, $f_{2}$, $f_{3}$, $f_{1}$ \\ 
$P^{SP^{\mu_{E_F},\mu_{E_W}}}_{f_{4}}:$ & $w_{4}$, $w_{2}$, $w_{3}$, $w_1$ &
$P^{SP^{\mu_{E_F},\mu_{E_W}}}_{w_{4}}:$ & $f_{3}$, $f_{1}$, $f_{4}$, $f_2$ \\ 
\end{tabular}
\end{center}
The core of the subprofile is
\[
C\big(SP^{\mu_{E_F},\mu_{E_W}}\big)=\{\mu_{E_F},\mu_F,\mu_1,\mu_W\},
\]
where
\[
\mu_{E_F}=
\begin{pmatrix}
f_1 & f_2 & f_3 & f_4\\
w_2 & w_3 & w_1 & w_4
\end{pmatrix}.
\]
(The step-by-step execution of the EADA mechanism, including the identification of the interrupter pair and the computation of $\mu_{E_F}$, is provided in Appendix \ref{apendice descripcion EADAM}.)

In particular, $\mu_{E_F}\notin C(P)$ but $\mu_{E_F}\in C\big(SP^{\mu_{E_F},\mu_{E_W}}\big)$. Hence,
\[
C(P)\subsetneq C\big(SP^{\mu_{E_F},\mu_{E_W}}\big).
\]
Moreover, since $\{f_3,w_2\}$ is the only blocking coalition of $\mu_{E_F}$ in $P$ and $(f_3,w_2)$ has been removed in $SP^{\mu_{E_F},\mu_{E_W}}$, no matching in $C(P)$ can dominate $\mu_{E_F}$. Therefore, $\mu_{E_F}$ is undominated in the reduced profile and $C\big(SP^{\mu_{E_F},\mu_{E_W}}\big)$ satisfies internal stability with respect to $P$.
\end{examplecont}

Besides identifying pairs to eliminate, the EADA mechanism generates two
sequences of matchings, depending on the proposing side.
The next proposition shows that each of these sequences is internally
stable with respect to the original profile.

\begin{proposition}\label{E_F y E_W cumple estabilidad interna}
Let $\mathcal{M}^{E_F}$ and $\mathcal{M}^{E_W}$ be the sets of matchings obtained when the EADA mechanism is executed with firms proposing and workers proposing, respectively. Then
\[
\mathcal{M}^{E_F},\, \mathcal{M}^{E_W} \in V_{IS}(P).
\]
\end{proposition}
\begin{proof}
The proof is relegated to Appendix \ref{apendice pruebas}.
\end{proof}

Furthermore, these matchings remain stable in the market restricted by the elimination of all interrupter pairs. The next proposition shows that both sequences are contained in the core of the subprofile induced by the EADA mechanism.

\begin{proposition}\label{EADAM subconjunto de P'}
Let $\mathcal{M}^{E_F}$ and $\mathcal{M}^{E_W}$ denote the sets of matchings obtained by the EADA mechanism when firms and workers propose, respectively. 
Then,
\[
\mathcal{M}^{E_F} \subseteq C(SP^{\mu_{E_F},\mu_{E_W}})
\qquad\text{and}\qquad
\mathcal{M}^{E_W} \subseteq C(SP^{\mu_{E_F},\mu_{E_W}}).
\]
\end{proposition}
\begin{proof}
The proof is relegated to Appendix \ref{apendice pruebas}.
\end{proof}

We now relate the core of the original profile with the core of the subprofile induced by the EADA mechanism. 
Since
$SP^{\mu_{E_F},\mu_{E_W}}\sqsubseteq P$
and
Remark~\ref{remark el conjunto de pares eliminados no contiene pares estables}
shows that
$\mathcal E^{EADA}$
contains no stable pair of $P$,
Lemma \ref{core es subconjunto de core de subperfil}
immediately yields

\[
C(P)\subseteq C(SP^{\mu_{E_F},\mu_{E_W}}).
\]
Applying
Lemma~\ref{todo core cumple estabilidad interna}
under the same assumptions, we also obtain
\[
C(SP^{\mu_{E_F},\mu_{E_W}})
\in
V_{IS}(P).
\]

Note that $C(P)$ and $C(SP^{\mu_{E_F},\mu_{E_W}})$ need not coincide. 
Whenever the inclusion is strict, the additional matchings in
$C(SP^{\mu_{E_F},\mu_{E_W}})$ become undominated after the interrupter
pairs are removed.
These matchings will play an important role in the construction of the
final reduced profile.

\subsubsection{Discussion: Compatibility between the two reductions}\label{subseccion de discusion de subperfiles}

The Irving--Leather and EADA reductions identify different classes of
pairs and therefore reveal different undominated matchings. We now discuss
how the two constructions interact. In particular, pairs eliminated by
either reduction cannot support dominance against the additional matchings
revealed by the other. This property will allow us to combine both
collections of eliminated pairs into a single reduced profile.

First, consider a pair
$(f,w)\in\mathcal{E}^{IL}$.
By Remark~\ref{remark el conjunto de pares eliminados para muF y muW no contiene pares estables},
every such pair satisfies either
\[
\mu_F(f) P_f w P_f \mu_W(f)
\qquad\text{and}\qquad
\mu_W(w) P_w f,
\]
or
\[
\mu_W(w) P_w f P_w \mu_F(w)
\qquad\text{and}\qquad
\mu_F(f) P_f w.
\]
Now let $\mu\in C(SP^{\mu_{E_F},\mu_{E_W}})$. By the bounds established
for the EADA-induced subprofile in
Appendix~\ref{apendice descripcion EADAM},
\[
\mu_{E_F}\succeq_F\mu\succeq_F\mu_F
\qquad\text{and}\qquad
\mu_{E_W}\succeq_W\mu\succeq_W\mu_W.
\]
In the first case, firm $f$ weakly prefers $\mu(f)$ to $\mu_F(f)$ and
$\mu_F(f)P_fw$, so $\mu(f)P_fw$. In the second case, worker $w$ weakly
prefers $\mu(w)$ to $\mu_W(w)$ and $\mu_W(w)P_wf$, so $\mu(w)P_wf$.
Hence, in either case, the coalition $\{f,w\}$ cannot support a dominance
relation against $\mu$. Thus, pairs eliminated by the Irving--Leather
reduction cannot support dominance against matchings in
$C(SP^{\mu_{E_F},\mu_{E_W}})$.

Conversely, consider a pair
$(f,w)\in\mathcal{E}^{EADA}$.
By Remark~\ref{remark el conjunto de pares eliminados no contiene pares estables},
either
\[
w P_f \mu_F(f)
\qquad\text{and}\qquad
\mu_F(w) P_w f,
\]
or
\[
f P_w \mu_W(w)
\qquad\text{and}\qquad
\mu_W(f) P_f w.
\]
For every $\mu\in C(SP^{\mu_F,\mu_W})$, the bounds determined by
$\mu_F$ and $\mu_W$ imply that, in the first case,
$\mu(w)P_wf$, whereas in the second case, $\mu(f)P_fw$.
Therefore, the coalition $\{f,w\}$ cannot support a dominance relation
against any matching in $C(SP^{\mu_F,\mu_W})$.

The two reductions are therefore compatible in the following sense:
the pairs eliminated by either procedure do not interfere with the
additional undominated matchings revealed by the other. Consequently,
the two collections of eliminated pairs can be combined into a single
reduced profile. The next subsection formalizes this construction.

\subsection{\texorpdfstring{The subprofile \(P^\star\)}{The subprofile P*}}\label{Definition of the profile P*}

The previous subsection introduced two complementary reductions of the
original preference profile. The first is based on the lattice structure
of the core, whereas the second arises from the efficiency-adjusted
improvements identified by the EADA mechanism. We now combine these two
sources of information into a single reduced profile, denoted by
$P^\star$.

As shown by the compatibility results established at the end of the
previous subsection, the pairs eliminated by one reduction do not affect
the additional undominated matchings identified by the other.
Consequently, both collections of eliminated pairs can be removed
simultaneously. The resulting subprofile will play a central role in the
remainder of the paper, as its core will be shown to coincide with the
unique vNM stable set.

\begin{definition}\label{definicion de subperfil P star}
The subprofile
$P^\star$
is obtained from
$P$
by eliminating every pair in
\[
\mathcal E^{IL}\cup\mathcal E^{EADA}.
\]
Equivalently,
\[
\mathcal E^{P\to P^\star}
=
\mathcal E^{IL}
\cup
\mathcal E^{EADA}.
\]
\end{definition}
Thus,
$P^\star$
is obtained by simultaneously applying the Irving--Leather and
EADA reductions.
This provides an explicit procedure to compute $P^\star$, obtained by aggregating the restrictions imposed by the two subprofiles.

We next relate the core of the combined subprofile $P^\star$ to the cores of the two subprofiles introduced earlier. 
Since
$P^\star$
contains every restriction imposed by the two previous subprofiles,
their cores are naturally preserved.
The next proposition formalizes this observation.
\begin{proposition}\label{proposition 1 de P*}
Given a profile $P$ and the associated subprofiles $SP^{\mu_{E_F},\mu_{E_W}}$, $SP^{\mu_F,\mu_W}$, and $P^\star$,  
we have that 
\[
C\!\left(SP^{\mu_{E_F},\mu_{E_W}}\right) \subseteq C(P^\star)
\quad \text{and} \quad 
C\!\left(SP^{\mu_F,\mu_W}\right) \subseteq C(P^\star).
\]
\end{proposition}
\begin{proof}
The proof is relegated to Appendix \ref{apendice pruebas}.
\end{proof}

We now examine the internal stability properties of the core of the combined subprofile. 
Since $P^\star$ is obtained by restricting the original profile without removing any stable pair of $P$, the core of this profile inherits internal stability when evaluated in the original market. 
Formally, we have the following result.

\begin{proposition}\label{core P^* satisfies internal stability}
Given the profile $P$ and the subprofile $P^\star$,  \(C(P^\star) \in V_{IS}(P)\).
\end{proposition}
\begin{proof}
The proof is relegated to Appendix \ref{apendice pruebas}.
\end{proof}

An important structural property of the subprofile $P^\star$ is that, for every firm $f \in F$, all pairs lying between $\mu_{E_F}(f)$ and $\mu_{E_W}(f)$ in the preference list $P^\star_f$ correspond to stable pairs of $P^\star$. 
The same holds symmetrically for each worker. 
The following example illustrates this feature.

\begin{examplecont}{ex:wako}
The two reductions eliminate the pairs $(f_2,w_1)$ and $(f_3,w_2)$. Hence,
\[
\mathcal{E}^{P \to P^\star}=\{(f_2,w_1),(f_3,w_2)\}.
\]
The resulting subprofile is given by
\begin{center}
\begin{tabular}{l l @{\hspace{1cm}} l l}
$P^\star_{f_{1}}:$ & $w_{2}$, $w_{4}$, $w_{1}$, $w_3$ &
$P^\star_{w_{1}}:$ & $f_{1}$, $f_{3}$, $f_4$ \\
$P^\star_{f_{2}}:$ & $w_{3}$, $w_{2}$, $w_4$ &
$P^\star_{w_{2}}:$ & $f_{2}$, $f_{4}$, $f_{1}$ \\
$P^\star_{f_{3}}:$ & $w_{3}$, $w_{1}$, $w_{4}$ &
$P^\star_{w_{3}}:$ & $f_{4}$, $f_{2}$, $f_{3}$, $f_1$ \\ 
$P^\star_{f_{4}}:$ & $w_{4}$, $w_{2}$, $w_{3}$, $w_1$ &
$P^\star_{w_{4}}:$ & $f_{3}$, $f_{1}$, $f_{4}$, $f_{2}$ \\ 
\end{tabular}
\end{center}
The core of the reduced profile is
\[
C\big(P^\star\big)=\{\mu_{E_F},\mu_F,\mu_1,\mu_W,\widetilde{\mu}\}.
\]
In particular,
\[
C(P)\subsetneq C\big(P^\star\big).
\]
As illustrated in this example and established by Proposition \ref{core P^* satisfies internal stability}, $C(P^\star)$ satisfies internal stability with respect to $P$.
\end{examplecont}

The previous results show that the core of the combined subprofile
$P^\star$ enlarges the original core while remaining internally stable
with respect to the original profile. The previous example illustrates these properties by exhibiting additional undominated matchings that arise
after combining the Irving--Leather and EADA reductions.

The next two propositions extend
Propositions~\ref{proposition mu no estable fuera de VI_s}
and~\ref{Proposition Resultado 2}
from the core of the original profile to the enlarged core
$C(P^\star)$.

The key observation is that, by construction of the EADA mechanism,
$\mu_{E_F}$ and $\mu_{E_W}$ bound the set $C(P^\star)$ from the perspective
of the two sides of the market. In particular, $\mu_{E_F}$ is
firm-optimal and $\mu_{E_W}$ is worker-optimal within
$C(P^\star)$. Therefore, for every $\mu\in C(P^\star)$,
\[
\mu_{E_F}\succeq_F\mu\succeq_F\mu_{E_W}
\qquad\text{and}\qquad
\mu_{E_W}\succeq_W\mu\succeq_W\mu_{E_F}.
\]
The next proposition is the analogue of Proposition~\ref{proposition mu no estable fuera de VI_s}, with $\mu_W$ replaced by $\mu_{E_W}$.

\begin{proposition}\label{result 1 for P^*}
Let $\mu \in \mathcal{IR}$. 
Assume that $F(\mu_{E_W}, \mu) \neq \emptyset$, and that no matching $\mu_1 \in \mathcal{IR}$ satisfies $\mu_1 \succeq_W \mu_{E_W}$. 
Then there is no set $V_{IS} \in V_{IS}(P)$ such that both $\mu$ and $\mu_{E_W}$ belong to $V_{IS}$.
\end{proposition}

\begin{proof}
The argument follows the same reasoning as in Proposition~\ref{proposition mu no estable fuera de VI_s}, applied to the bounds induced by $\mu_{E_W}$.
\end{proof}

By symmetry, if $W(\mu_{E_F},\mu)\neq\emptyset$ and no matching $\mu_2\in\mathcal{IR}$ satisfies $\mu_2\succeq_F\mu_{E_F}$, then there is no set $V_{IS}\in V_{IS}(P)$ such that both $\mu$ and $\mu_{E_F}$ belong to $V_{IS}$.

The next proposition is the analogue of
Proposition~\ref{Proposition Resultado 2},
where the core of the original profile is replaced by the enlarged core
$C(P^\star)$.
\begin{proposition}\label{result 2 for P^*}
Let $\mu \in \mathcal{IR}$ and let $(f,w)$ be such that $\mu(f)=w$. 
If $(f,w)$ is not a stable pair of $P^\star$ and 
\[
\mu_{E_F}(f) P_f \mu(f) P_f \mu_{E_W}(f),
\]
then there is no set $V_{IS} \in V_{IS}(P)$ such that 
\[
\{\mu\}\cup C(P^\star) \subseteq V_{IS}.
\]
\end{proposition}

\begin{proof}
The proof follows the same argument as Proposition~\ref{Proposition Resultado 2}, replacing $C(P)$ with $C(P^\star)$.
\end{proof}

\subsection{Characterization of the von Neumann--Morgenstern stable set}\label{seccion caracterizacion de VNM}

We now turn to the main result of the paper, which brings together all the
structural ingredients developed in the previous sections.

To develop our analysis, we first generalized the classical Decomposition Lemma and introduced IS-cycles to describe the structure of internally stable sets. We then constructed two complementary reductions of the preference profile: one based on the lattice structure of the core through the Irving--Leather reduction, and another based on the efficiency improvements identified by the EADA mechanism. Combining these two sources of information led to the reduced profile $P^\star$.

Moreover, the construction of $P^\star$ removes no stable pair of the
original profile. Consequently,
\[
C(P)\subseteq C(P^\star).
\]
Since $C(P)\neq\emptyset$, it follows that $C(P^\star)\neq\emptyset$.
The previous sections also established that $C(P^\star)$ contains all
matchings revealed by the two reduction procedures and remains internally
stable with respect to the original profile.

The theorem below shows that these structural results are sufficient to
completely characterize the vNM stable set.
More importantly, they transform the problem of computing a cooperative
solution concept defined through dominance relations into the classical
problem of computing the core of an explicitly constructed reduced
preference profile.

The only remaining step is to establish external stability. Once this is
proved, it follows that $C(P^\star)$ coincides with the unique
vNM stable set.

\begin{theorem}\label{main theorem}
Let $P$ be a one-to-one matching market, and let $P^\star$ be the reduced
preference profile constructed in
Definition~\ref{definicion de subperfil P star}. Then,
the unique vNM stable set of $P$ is given by the core
of $P^\star$. Equivalently,
\[
\mathrm{vNM}(P)=\{\,C(P^\star)\,\}.
\]
\end{theorem}
\begin{proof}
By Proposition~\ref{core P^* satisfies internal stability}, we have that 
\(C(P^\star)\in V_{IS}(P)\). 
We first need to show that \(C(P^\star)\in V_{ES}(P)\). 
Suppose, by contradiction, that $C(P^\star)$ does not satisfy external stability. Then,
there exists a matching $\mu\in IR\setminus C(P^\star)$ such that no matching
$\widehat{\mu}\in C(P^\star)$ and no coalition $S$ satisfy
$\widehat{\mu}\succ_S\mu$. We claim that the enlarged set
$\{\mu\}\cup C(P^\star)$ satisfies internal stability.

By Proposition~\ref{core P^* satisfies internal stability}, no two matchings in $C(P^\star)$ dominate each other.
Moreover, by the choice of $\mu$, no matching in $C(P^\star)$ dominates $\mu$.
It only remains to prove that $\mu$ does not dominate any matching in
$C(P^\star)$.

Suppose, to the contrary, that there exist a matching
$\mu'\in C(P^\star)$ and a coalition $S$ such that
$\mu\succ_S\mu'$.
If $\mu'\in C(P)$, this contradicts the definition of the core.
Hence, $\mu'\in C(P^\star)\setminus C(P)$.
By Lemma \ref{Lema para unicidad}, there exist a matching
$\overline{\mu}\in C(P)$ and a coalition $S'$ such that
$\overline{\mu}\succ_{S'}\mu$.
Since $C(P)\subseteq C(P^\star)$, this contradicts the choice of $\mu$.
Therefore, $\{\mu\}\cup C(P^\star)\in V_{IS}(P)$.

Since \(\mu \notin C(P^\star)\), we must distinguish two cases. 
First, all pairs matched by \(\mu\) are stable with respect to \(P^\star\); in this case, the blocking coalition \(S\) consists of an unmatched pair under $\mu$. 
Second, \(\mu\) assigns at least one unstable pair with respect $P^\star$; in this case, the blocking coalition \(S\) consists of a pair that is unstable under 
\(P^\star\).

\begin{description}
    \item[\textbf{Case 1: all pairs matched by \(\boldsymbol{\mu}\)  are stable with respect to \(\boldsymbol{P^\star}\).} ]
Consider \(S=\{f,w\}\), where \((f,w)\) is an unmatched pair under $\mu$. 
If \((f,w)\) is a stable pair of \(P^\star\), it is straightforward that there is $\mu'\in C(P^\star)$ such that
\(\mu' \succ_{\{f,w\}} \mu\).
If instead \((f,w)\) is unstable with respect to \(P^\star\), and since $S$ is a blocking coalition, we have
\[
w \; P^\star_f \; \mu(f)
\qquad\text{and}\qquad
f \; P^\star_w \; \mu(w).
\]
By the hypothesis of this case,  for each pair 
\((f',w')\) satisfying \(\mu(f')=w'\), there is \(\mu' \in C(P^\star)\) such that \(\mu'(f') = w'\). 
Hence,
\[
w \; P^\star_f \; \mu(f) = \mu'(f)
\qquad\text{and}\qquad
f \; P^\star_w \; \mu(w) = \mu'(w),
\]
which implies that \(S=\{f,w\}\) is also a blocking coalition for \(\mu'\). 
This is a contradiction, since \(\mu'\in C(P^\star)\).

\item[\textbf{Case 2: $\boldsymbol{\mu$ assigns at least one unstable pair with respect to $P^\star}$.}] That is, there is a pair \((f,w)\) that is not stable under \(P^\star\) and for which \(\mu(f) = w\).
Since every matching in
$C(P^\star)$
lies between
$\mu_{E_F}$
and
$\mu_{E_W}$,
every unstable matched pair must fall into exactly one of the following
three situations.
  \begin{description}
      \item[\text{Subcase 1: \(\boldsymbol{w = \mu(f) \; P_f \; \mu_{E_F}(f)}\).}]
Thus, \(F(\mu,\mu_{E_F}) \neq \emptyset\) in \(P\). Since  $\mu, \mu_{E_F}\in\{\mu\}\cup C(P^\star)$,  by the Decomposition Lemma, we have 
\(W(\mu_{E_F},\mu) \neq \emptyset\) in \(P\). 
By the Pareto efficiency properties of the EADA mechanism established by \cite{kesten2010school}, and recalled in Appendix~\ref{apendice descripcion EADAM}, there is no matching \(\widetilde{\mu} \in \mathcal{IR}\) such that $\widetilde{\mu} \succeq_F \mu_{E_F}$. 
Hence, by Proposition~\ref{result 1 for P^*}, there is no \(V_{IS}\) containing both \(\mu\) and \(\mu_{E_F}\). 
 
\item[Subcase 2: \(\boldsymbol{\mu_{E_F}(f) \; P_f \; w = \mu(f) \; P_f \; \mu_{E_W}(f)}\).]
Since $(f,w)$ is not a stable pair of $P^*$, the hypotheses of Proposition~\ref{result 2 for P^*} are satisfied. Therefore, there is no set $V_{IS} \in V_{IS}(P)$ such that $\{\mu\} \cup C(P^*) \subseteq V_{IS}$.
Therefore, there is no \(V_{IS}\) such that 
\(\{\mu\}\cup C(P^\star) \subseteq V_{IS}\). 

\item[\text{Subcase 3: $\boldsymbol{\mu_{E_W}(f) \;  P_f \; \mu(f) = w}$.}] Thus, \(F(\mu_{E_W}, \mu) \neq \emptyset\) in \(P\). By the Pareto efficiency properties of the EADA mechanism established by \cite{kesten2010school}, and recalled in Appendix~\ref{apendice descripcion EADAM}, there is no matching \(\widetilde{\mu} \in \mathcal{IR}\) such that $\widetilde{\mu} \succeq_W \mu_{E_W}$.
Hence, by Proposition~\ref{result 1 for P^*}, there is no \(V_{IS}\) containing both \(\mu\) and \(\mu_{E_W}\). 
 \end{description}
In all three subcases considered, we obtain that there is no set 
\(V_{IS} \in V_{IS}(P)\) that contains a matching 
\(\mu \in \mathcal{IR}\setminus C(P^\star)\) together with the matchings in 
\(C(P^\star)\). 
Therefore, the contradiction arises from the fact that the enlarged set 
\(\{\mu\}\cup C(P^\star)\) satisfies internal stability.
\end{description}
Therefore, by analyzing these two cases, $C(P^\star)$ satisfies external stability with respect to $P$, therefore $C(P^\star) \in \text{vNM}(P)$.

We now prove that no other internally and externally stable set can exist.
To prove uniqueness, we proceed by contradiction. Assume that there is another stable set $V$ for the market $P$.
By Remark~\ref{remark de core en von neuman}, we have
$C(P)\subseteq V.$

To prove uniqueness, let $V\in \text{vNM}(P)$ and suppose that
$V\neq C(P^\star)$.
We proceed in two steps.
First, we show that no stable set can omit a matching belonging to
$C(P^\star)$.
Second, we prove that no stable set can contain a matching outside
$C(P^\star)$.

We first rule out the possibility that another stable set omits some
matching belonging to $C(P^\star)$.
Suppose, to the contrary, that
$V\subsetneq C(P^\star)$.
Then there exists a matching
$\mu'\in C(P^\star)\setminus V$.
Since $V$ satisfies external stability, there exist
$\mu\in V$ and a coalition $S$ such that
$\mu\succ_S\mu'$.
However, since
$V\subseteq C(P^\star)$,
both $\mu$ and $\mu'$ belong to
$C(P^\star)$.
This contradicts
Proposition~\ref{core P^* satisfies internal stability}, which states that
$C(P^\star)$ is internally stable with respect to $P$.
Hence,
$V$ cannot be a proper subset of
$C(P^\star)$.

We next rule out the possibility that a stable set contains a matching outside $C(P^\star)$. Since $V\neq C(P^\star)$ and we have already established that $V$ cannot be a proper subset of $C(P^\star)$, there must exist a matching
$\widetilde{\mu}\in V\setminus C(P^\star)$.

Since $C(P^\star)$ satisfies external stability, there is a matching
$\mu'\in C(P^\star)$ and a coalition $S$ such that $\mu'\succ_S\widetilde{\mu}.$ Because $\widetilde{\mu}\in V$ and $V$ satisfies internal stability, we
must have $\mu'\notin V.$

By the external stability of $V$, there is a matching
$\widehat{\mu}\in V$ and a coalition $T$ such that $\widehat{\mu}\succ_T\mu'.$
Since $\mu'\in C(P^\star)$ and $C(P^\star)$ satisfies internal stability,
it follows that
\[
\widehat{\mu}\notin C(P^\star).
\]

Moreover, $\mu'\notin C(P)$, since no matching in the core of $P$ can be
dominated under the original preference profile. Therefore,
\[
\mu'\in C(P^\star)\setminus C(P).
\]

Applying Lemma~\ref{Lema para unicidad} to
$\widehat{\mu}\succ_T\mu'$, there exist a matching
$\overline{\mu}\in C(P)$ and a coalition $T'$ such that $\overline{\mu}\succ_{T'}\widehat{\mu}.$
Since $C(P)\subseteq V$, we have $\overline{\mu},\widehat{\mu}\in V.$
Thus, two matchings belonging to $V$ satisfy
\[
\overline{\mu}\succ_{T'}\widehat{\mu},
\]
which contradicts the internal stability of $V$.

Therefore, no vNM stable set distinct from
$C(P^\star)$ exists. Hence, $C(P^\star)$ is the unique
vNM stable set of the market.
\end{proof}

We next discuss how the characterization provided by Theorem~\ref{main theorem} specializes under different configurations of the market. 
These cases clarify when the unique vNM stable set coincides with the core of the original profile or with the core of one of the intermediate subprofiles. 
We organize the discussion according to the different configurations that may arise.

\begin{enumerate}[(i)]
    \item \textbf{When the EADA mechanism coincides with the core bounds and no pairs are removed:}  
If $\mu_{E_F}=\mu_F$, $\mu_{E_W}=\mu_W$, and $SP^{\mu_F,\mu_W}=P$, then no pair is eliminated in the construction of either subprofile. 
Consequently, $P^\star=P$, so the core of the original profile already captures all dominance–relevant outcomes. 
In this case, Theorem~\ref{main theorem} implies that the unique vNM stable set is precisely $\{C(P)\}$.

\medskip

\item \textbf{When only the core-based restriction is active:}  
Suppose that $\mu_{E_F}=\mu_F$ and $\mu_{E_W}=\mu_W$, but $SP^{\mu_F,\mu_W}\neq P$. 
In this case, the EADA mechanism does not generate additional improvements, while the restriction induced by the extremal core matchings removes pairs that cannot arise along core transitions. 
Thus $P^\star$ coincides with $SP^{\mu_F,\mu_W}$, and Theorem~\ref{main theorem} implies that the unique stable set of the market is given by $\{C(SP^{\mu_F,\mu_W})\}$.

\medskip

\item \textbf{When the EADA restriction is the binding one.}  
Finally, suppose that $\mu_{E_F}\neq \mu_F$ or $\mu_{E_W}\neq \mu_W$, while $SP^{\mu_F,\mu_W}=P$. 
In this environment, the only restrictions arise from the efficiency–adjusted improvements identified by the EADA mechanism. 
Thus $P^\star$ coincides with $SP^{\mu_{E_F},\mu_{E_W}}$, and Theorem~\ref{main theorem} implies that the unique stable set is $\{C(SP^{\mu_{E_F},\mu_{E_W}})\}$.

\end{enumerate}
These three situations illustrate how the characterization provided by
Theorem~\ref{main theorem} continuously interpolates between the classical
core characterization and the enlarged cores generated by the two
reduction procedures.

\section{A short path to the core}\label{seccion path to core}

In decentralized matching markets, agents may initially be organized under a matching that does not belong to the core. In such situations, a natural question arises: can agents, by successively reorganizing themselves through coalitional deviations, eventually reach a core matching without the intervention of a central authority?

This adjustment process is commonly referred to in the literature as a \textbf{path to stability}. In our framework, it captures the idea that, starting from an arbitrary matching, agents may move toward a core matching through a sequence of coalitional dominations.

Formally, given two matchings $\mu$ and $\mu'$, a \textbf{path} from $\mu$ to $\mu'$ is a sequence of matchings $\mu_1,\ldots,\mu_k$ and coalitions $S_1,\ldots,S_{k-1}$ such that
\[
\mu'=\mu_k \succ_{S_{k-1}} \mu_{k-1} \succ_{S_{k-2}} \cdots \succ_{S_1} \mu_1=\mu.
\]
Thus, each matching along the path dominates its predecessor through a
coalitional deviation.

The characterization obtained in the previous section provides a natural
route for such an adjustment process. Since $C(P^\star)$ is the unique
vNM stable set, every matching outside $C(P^\star)$
is dominated by some matching in $C(P^\star)$. Moreover,
Lemma~\ref{Lema para unicidad} shows that whenever a matching in
$C(P^\star)\setminus C(P)$ is dominated from outside $C(P^\star)$, the
dominating matching is itself dominated by some matching in the original
core. These two observations generate a path from any initial matching to
$C(P)$.

The next result shows that this path is remarkably short: from any
matching outside $C(P)$, a core matching can be reached through at most
three coalitional dominance steps.

\begin{proposition}[A short path to the core]
For every profile $P$, any matching $\mu\notin C(P)$ admits a path to
some matching in $C(P)$ of length at most three.
\end{proposition}
\begin{proof}
Let $\mu\notin C(P)$. We show that there is always a path from $\mu$ to
a core matching. The argument proceeds by considering two cases.

\medskip

\begin{description}

\item[\textbf{Case 1:} $\mu\in C(P^\star)\setminus C(P)$.]

Set $\mu_1=\mu$. Since $\mu_1\notin C(P)$, there exist a matching
$\mu_2$ and a coalition $S_1$ such that
\begin{equation}\label{ecu 1 propo path}
    \mu_2\succ_{S_1}\mu_1
\end{equation}
under $P$.

Since $C(P^\star)$ is internally stable with respect to $P$ and
$\mu_1\in C(P^\star)$, it follows that
$\mu_2\notin C(P^\star)$.
Applying Lemma~\ref{Lema para unicidad}, there exist a matching
$\mu_3\in C(P)$ and a coalition $S_2$ such that
\begin{equation}\label{ecu 2 propo path}
    \mu_3\succ_{S_2}\mu_2.
\end{equation}
Combining \eqref{ecu 1 propo path} and \eqref{ecu 2 propo path}, we obtain
\[
\mu_3\succ_{S_2}\mu_2\succ_{S_1}\mu_1=\mu.
\]
Thus, a core matching is reached in two dominance steps.

\medskip

\item[\textbf{Case 2:} $\mu\notin C(P^\star)$.]

Set $\mu_1=\mu$. Since $C(P^\star)$ is a vNM
stable set, it satisfies external stability. Hence, there exist
$\mu_2\in C(P^\star)$ and a coalition $S_1$ such that
\begin{equation}\label{ecu 3 propo path}
    \mu_2\succ_{S_1}\mu_1.
\end{equation}

If $\mu_2\in C(P)$, the path has length one and we are done.
Otherwise,
\[
\mu_2\in C(P^\star)\setminus C(P),
\]
so Case~1 applies to $\mu_2$. Hence, there exist a matching
$\mu_3\notin C(P^\star)$, a matching $\mu_4\in C(P)$, and coalitions
$S_2$ and $S_3$ such that
\begin{equation}\label{ecu 4 propo path}
    \mu_4\succ_{S_3}\mu_3\succ_{S_2}\mu_2.
\end{equation}
Combining \eqref{ecu 3 propo path} and \eqref{ecu 4 propo path}, we obtain
\[
\mu_4\succ_{S_3}\mu_3\succ_{S_2}\mu_2\succ_{S_1}\mu_1=\mu.
\]
Thus, a core matching is reached in three dominance steps.

\end{description}

In both cases, starting from any matching outside $C(P)$, there is a path
of length at most three leading to a matching in $C(P)$. The three-step
bound reflects the two-layer structure uncovered by the characterization
theorem: external stability provides a first step into $C(P^\star)$ when
needed, while at most two further dominance steps lead from
$C(P^\star)\setminus C(P)$ to the original core.
\end{proof}




    

\section{Concluding Remarks}\label{conclusiones}

This paper provides a structural characterization of the von
Neumann--Morgenstern stable set in one-to-one matching markets with strict
preferences. Our main result shows that the vNM stable set is unique and
coincides with the core of an explicitly constructed reduced preference
profile $P^\star$. Thus, a cooperative solution concept defined through
dominance relations among sets of matchings can be represented, and
computed, through a classical core problem in a suitably reduced matching
market.

The construction of $P^\star$ brings together the different structural
ingredients developed throughout the paper. A first ingredient is the
generalization of the classical Decomposition Lemma from core matchings to
arbitrary internally stable sets. The resulting IS-cycles reveal that
differences between matchings belonging to the same internally stable set
are organized through disjoint cyclic structures. Besides its role in the
characterization of the vNM stable set, this decomposition provides a
structural description of internally stable sets that may be of independent
interest.

Building on this structure, we identify two complementary sources of
information about the dominance relations relevant for vNM stability. The
first is provided by the lattice structure of the core and the
Irving--Leather reduction, which identifies non-stable pairs that become
irrelevant when moving across core outcomes. The second comes from the EADA
mechanism and identifies interrupter pairs associated with
efficiency-improving reallocations. Neither reduction alone captures the
entire structure of the vNM stable set. Their combination, however, leads
to the reduced profile $P^\star$.

The main characterization shows that the core of this reduced profile is
both internally and externally stable with respect to the original market
and therefore coincides with the unique vNM stable set. In this sense, the
role of the reduction is not merely to simplify the original preference
profile. Rather, it isolates the pairs that are relevant for the
dominance structure underlying vNM stability and transforms the original
stable-set problem into a standard core problem. Consequently, the
characterization provides both a structural interpretation and a
constructive procedure for determining the unique vNM stable set using
standard tools from matching theory.

This approach also provides a different perspective on the uniqueness
result established by \cite{wako2010polynomial}. Rather than relying on an
iterative modification of preferences based on successive sets of
undominated matchings, our argument derives the reduced market from the
structural properties of internally stable sets and from the two
complementary reduction procedures. Thus, both the construction of
$P^\star$ and the proof of the characterization are obtained directly
from the dominance structure of the original matching market.

The characterization also has implications for the dynamics of stability.
The reduced core $C(P^\star)$ provides an intermediate layer between
arbitrary matchings and the original core $C(P)$. External stability first
allows a matching outside $C(P^\star)$ to move into the reduced core,
while the structural properties underlying the characterization imply that
a matching in $C(P^\star)\setminus C(P)$ can reach the original core in at
most two additional dominance steps. As a consequence, every matching
outside $C(P)$ admits a path to a core matching of length at most three.
Hence, the same structure that characterizes the vNM stable set also
provides a simple description of how decentralized coalitional deviations
can lead to core outcomes.

Several questions remain open for future research. One natural direction
is to further investigate the scope of the generalized Decomposition Lemma.
Since the result applies to arbitrary internally stable sets rather than
only to core matchings, the associated IS-cycle structure may be useful for
studying other solution concepts defined through dominance relations.

Another important direction is to investigate whether the reduction
approach developed here extends beyond one-to-one matching markets. In
particular, it would be interesting to determine whether analogous
lattice-based and efficiency-based reductions can be constructed in
many-to-one or many-to-many environments, and whether their combination
again yields a reduced market whose core represents the corresponding vNM
stable set. More broadly, these questions may help clarify the extent to
which dominance-based cooperative solution concepts can be understood
through the classical structural tools of matching theory.

\appendix
\section{Appendix: Proofs}\label{apendice pruebas}

\noindent\begin{proof}[Proof of Lemma \ref{lema descomposicion para estables}]Let $\mu, \mu' \in V_{IS}$. First, we prove that $\mu'$ maps from $F(\mu',\mu)$ into $W(\mu,\mu')$. To do so, we show that $\mu'(F(\mu', \mu)) \subseteq W(\mu, \mu')$. Let $w\in \mu'(F(\mu', \mu))$, so there is $f\in F(\mu',\mu)$ such that $w=\mu'(f)$ and $w=\mu'(f) P_f \mu(f)$.  If $w \notin W(\mu, \mu')$, i.e., $\mu'(w) P_w \mu(w)$, then together with $f \in F(\mu',\mu)$, it would imply $\mu' \succ_{\{ f,w \}} \mu$, contradicting internal stability. Hence, $w \in W(\mu, \mu')$. Thus, $\mu'(F(\mu', \mu)) \subseteq W(\mu, \mu')$. Moreover, since $\mu'$ is injective, we have
\begin{equation}\label{1}
    | F(\mu', \mu)|\leq |  W(\mu, \mu')|.
\end{equation}

Second, we prove that $\mu$ maps from $W(\mu,\mu')$ into $F(\mu',\mu)$. To do so, we show that $\mu(W(\mu,\mu')) \subseteq F(\mu',\mu)$. Let $f'\in \mu(W(\mu, \mu'))$, so there is $w'\in W(\mu,\mu')$ such that $f'=\mu(w')$ and $f'=\mu(w') P_ {w'} \mu'(w')$. If $f' \notin F(\mu', \mu)$, i.e., $\mu(f') P_{f'} \mu'(f')$, then together with $w' \in W(\mu,\mu')$, it would imply $\mu \succ_{\{ f',w' \}} \mu'$, contradicting internal stability. Hence, $f' \in F(\mu', \mu)$. Thus, $\mu(W(\mu, \mu')) \subseteq F(\mu', \mu)$. Moreover, since $\mu$ is injective, we have
\begin{equation}\label{2}
    | W(\mu, \mu')|\leq |  F(\mu', \mu)|.
\end{equation}
Then, by (\ref{1}) and (\ref{2}), we conclude that $| W(\mu, \mu')| = |  F(\mu', \mu)|$.

Finally, since $\mu$ and $\mu'$ are injective, and the domain and codomain are finite and of the same cardinality, we have that $\mu$ and $\mu'$ map $F(\mu',\mu)$ onto $W(\mu,\mu')$ and $F(\mu,\mu')$ onto $W(\mu',\mu)$. That is, $\mu$ and $\mu'$ are bijective with the respective domains and codomains.
\end{proof}

\noindent\begin{proof}[Proof of Proposition \ref{existencia de un IS-cycle}]
We now prove that whenever $F(\mu', \mu) \neq \emptyset$, there is a cycle. 
To show this, consider a directed graph defined as follows. 
For each $f \in F(\mu', \mu)$, there is an arc pointing to the worker $w \in W$ such that $\mu'(f) = w$. 
By the Decomposition Lemma, we have that $w \in W(\mu, \mu')$. 
Hence, there is $f' \in F(\mu', \mu)$ such that $f' = \mu(w)$. 
Consequently, the digraph contains an arc from $w$ to $f'$.

We claim that this digraph eventually closes, forming the desired cycle. 
Suppose, to the contrary, that this is not the case. 
Then there is some worker $w_{i-1}$ that is reached by an arc but from which no arc departs. 
This would imply that there is no $f_i \in F(\mu', \mu)$ such that $w_{i-1} = \mu(f_i)$, 
contradicting the Decomposition Lemma. 
Similarly, suppose that there is some $f_i \in F(\mu', \mu)$ from which no arc departs. 
This would mean that there is no $w_i \in W(\mu, \mu')$ such that $\mu'(f_i) = w_i$, 
again contradicting the Decomposition Lemma. 
Therefore, the digraph must eventually close, forming a IS-cycle  $O^{\mu',\mu}$.
\end{proof}

\noindent\begin{proof}[Proof of Theorem~\ref{Teorema descomposicion ciclos disjuntos}]
Let $V_{IS}\in V_{IS}(P)$ and let $\mu',\mu\in V_{IS}$.
By Proposition~\ref{existencia de un IS-cycle}, if
$F(\mu',\mu)\neq\emptyset$, there is an IS-cycle
$O_1^{\mu',\mu}$ such that
\[
O_1^{\mu',\mu}\cap F\subseteq F(\mu',\mu)
\qquad\text{and}\qquad
O_1^{\mu',\mu}\cap W\subseteq W(\mu,\mu').
\]

If $
O_1^{\mu',\mu}\cap F=F(\mu',\mu),
$ there is nothing to prove. Otherwise, choose
\[
f\in F(\mu',\mu)\setminus O_1^{\mu',\mu}.
\]
Following the construction in the proof of
Proposition~\ref{existencia de un IS-cycle}, set
$w=\mu'(f)$. By the Decomposition Lemma,
$w\in W(\mu,\mu')$.
Moreover, $w\notin O_1^{\mu',\mu}$. Indeed, if
$w\in O_1^{\mu',\mu}$, then there would exist a firm
$\widehat f\in O_1^{\mu',\mu}$ such that
$\mu'(\widehat f)=w=\mu'(f)$.
Since $\mu'$ is injective, this would imply
$\widehat f=f$, contradicting
$f\notin O_1^{\mu',\mu}$.

By the Decomposition Lemma, there is
$f'\in F(\mu',\mu)$ such that
\[
\mu(f')=w.
\]
The same argument, now using the injectivity of $\mu$, implies that
$f'\notin O_1^{\mu',\mu}$.
Repeating this construction produces an alternating sequence entirely
contained in
\[
\bigl(F(\mu',\mu)\setminus O_1^{\mu',\mu}\bigr)
\cup
\bigl(W(\mu,\mu')\setminus O_1^{\mu',\mu}\bigr).
\]
Since the market is finite, this sequence eventually closes and yields
another IS-cycle, denoted by $O_2^{\mu',\mu}$, which is disjoint from
$O_1^{\mu',\mu}$.

If
\[
\left(O_1^{\mu',\mu}\cup O_2^{\mu',\mu}\right)\cap F
=
F(\mu',\mu),
\]
we are done. Otherwise, we repeat the same argument starting from a firm
in
\[
F(\mu',\mu)\setminus
\left(O_1^{\mu',\mu}\cup O_2^{\mu',\mu}\right).
\]
At each step, we obtain an IS-cycle disjoint from all previously
constructed cycles.

Since $F(\mu',\mu)$ is finite and each new cycle contains at least one
previously uncovered firm, the procedure terminates after finitely many
steps. Hence, there exist pairwise disjoint IS-cycles
$O_1^{\mu',\mu},\ldots,O_k^{\mu',\mu}$ such that
\[
\left[\bigcup_{i=1}^{k}O_i^{\mu',\mu}\right]\cap F
=
F(\mu',\mu).
\]
Thus, every firm in $F(\mu',\mu)$ belongs to an IS-cycle between
$\mu'$ and $\mu$.
\end{proof}

\noindent\begin{proof}[Proof of Proposition~\ref{proposition mu no estable fuera de VI_s}]
Suppose, by contradiction, that there exists
$V_{IS}\in V_{IS}(P)$ such that
$\mu,\mu_W\in V_{IS}$.

Since $F(\mu_W,\mu)\neq\emptyset$, Theorem~
\ref{Teorema descomposicion ciclos disjuntos} yields a collection of
pairwise disjoint IS-cycles
$O_1^{\mu_W,\mu},\ldots,O_k^{\mu_W,\mu}$ such that
\[
\left[\bigcup_{i=1}^{k}O_i^{\mu_W,\mu}\right]\cap F
=
F(\mu_W,\mu).
\]
Let
\[
U=\bigcup_{i=1}^{k}O_i^{\mu_W,\mu}.
\]

Define a matching $\overline{\mu}$ by
\[
\overline{\mu}(a)=
\begin{cases}
\mu(a), & \text{if } a\in U,\\
\mu_W(a), & \text{if } a\notin U.
\end{cases}
\]

We first observe that $\overline{\mu}$ is well defined as a matching.
Indeed, by the definition of an IS-cycle, each cycle
$O_i^{\mu_W,\mu}$ is closed under both $\mu$ and $\mu_W$: the firms
belonging to the cycle are matched, under either matching, to workers
belonging to the same cycle, and conversely. Since the cycles are pairwise
disjoint, replacing $\mu_W$ by $\mu$ on their union preserves the matching
property.

Moreover, $\overline{\mu}\in\mathcal{IR}$. Every pair used by
$\overline{\mu}$ is a pair matched either under $\mu$ or under $\mu_W$,
and both matchings are individually rational.

We now compare $\overline{\mu}$ and $\mu_W$ from the workers' side.
For every
\[
w\in U\cap W,
\]
the definition of an IS-cycle implies that
$w\in W(\mu,\mu_W)$. Hence,
\[
\overline{\mu}(w)=\mu(w)P_w\mu_W(w).
\]
On the other hand, if $w\notin U$, then
\[
\overline{\mu}(w)=\mu_W(w).
\]
Therefore,
\[
\overline{\mu}\succeq_W\mu_W.
\]

Since $F(\mu_W,\mu)\neq\emptyset$, at least one IS-cycle is nonempty,
and hence at least one worker strictly prefers $\overline{\mu}$ to
$\mu_W$. Thus,
$\overline{\mu}\in\mathcal{IR}$ satisfies
$\overline{\mu}\succeq_W\mu_W$, contradicting the hypothesis that no
matching $\mu_1\in\mathcal{IR}$ satisfies
$\mu_1\succeq_W\mu_W$.

Therefore, no internally stable set can contain both $\mu$ and $\mu_W$.
\end{proof}

\noindent\begin{proof}[Proof of Proposition~\ref{Proposition Resultado 2}]
Let $\mu\in\mathcal{IR}$ and let $(f,w)$ satisfy $\mu(f)=w$.
Suppose, by contradiction, that there is
$V_{IS}\in V_{IS}(P)$ such that
\[
\{\mu\}\cup C(P)\subseteq V_{IS}.
\]
Since $(f,w)$ is not a stable pair of $P$, no matching in $C(P)$ matches
$f$ with $w$. By hypothesis,
\[
\mu_F(f)P_f\mu(f)P_f\mu_W(f).
\]

Because $\mu_W$ and $\mu_F$ are, respectively, the worker-optimal and
worker-pessimal core matchings, and $(f,w)$ is not a stable pair, exactly
one of the following three cases must hold.

\begin{description}

\item[\textbf{Case 1:}
$\boldsymbol{\mu(w)P_w\mu_W(w)}$.]

Since
\[
\mu(f)P_f\mu_W(f)
\]
and
\[
\mu(w)P_w\mu_W(w),
\]
we have
\[
\mu\succ_{\{f,w\}}\mu_W.
\]
Since both $\mu$ and $\mu_W$ belong to $V_{IS}$, this contradicts the
internal stability of $V_{IS}$.

\item[\textbf{Case 2:}
$\boldsymbol{\mu_F(w)P_w\mu(w)}$.]

By hypothesis,
\[
\mu_F(f)P_f\mu(f),
\]
so $f\in F(\mu_F,\mu)$. Since
$\mu_F,\mu\in V_{IS}$, the Decomposition Lemma implies
\[
\mu(f)=w\in W(\mu,\mu_F).
\]
Therefore,
\[
\mu(w)P_w\mu_F(w),
\]
contradicting the assumption of this case.

\item[\textbf{Case 3:}
$\boldsymbol{\mu_W(w)P_w\mu(w)P_w\mu_F(w)}$.]

We first establish the following claim.

\begin{description}
\item[\textbf{Claim.}]
There is a core matching $\widehat{\mu}\in C(P)$ such that
$(f,w)$ ceases to be mutually acceptable at the Irving--Leather reduction
associated with $\widehat{\mu}$ and $\mu_W$.
\end{description}

To see this, recall that $(f,w)$ is mutually acceptable under the original
profile $P$. Moreover,
\[
\mu_F(f)P_f wP_f\mu_W(f)
\]
and
\[
\mu_W(w)P_w fP_w\mu_F(w).
\]
Since $(f,w)$ is not a stable pair of $P$, it must be eliminated during the
Irving--Leather procedure that generates the lattice of core matchings.
Hence, there is a core matching
$\widehat{\mu}\in C(P)$ at whose corresponding reduction step the pair
$(f,w)$ ceases to be mutually acceptable. This proves the claim.

Since $(f,w)$ survives until the reduction associated with
$\widehat{\mu}$, there is a predecessor core matching
$\widetilde{\mu}\in C(P)$ and an Irving--Leather cycle $\sigma$ such that
\[
\widehat{\mu}=\widetilde{\mu}(\sigma),
\]
while $(f,w)$ remains mutually acceptable immediately before the
elimination step.

We claim that
\begin{equation}\label{ecu 1 prueba propo 4}
\widehat{\mu}(f)P_f\mu(f).
\end{equation}
Suppose otherwise that
\[
\mu(f)P_f\widehat{\mu}(f).
\]
Since $(f,w)$ is mutually acceptable before the cycle is eliminated and
$\widehat{\mu}=\widetilde{\mu}(\sigma)$, we would obtain
\[
\widetilde{\mu}(f)
P_f
\mu(f)
P_f
\widehat{\mu}(f),
\]
with $\mu(f)$ lying strictly between the two consecutive partners of $f$
determined by the cycle. This contradicts the construction of the
Irving--Leather reduction. Hence,
\eqref{ecu 1 prueba propo 4} holds.

Since $(f,w)$ is no longer mutually acceptable after the reduction step
associated with $\widehat{\mu}$, while
\eqref{ecu 1 prueba propo 4} holds, the deletion must occur on the worker
side. Therefore,
\[
\widehat{\mu}(w)P_w\mu(w).
\]

On the other hand, \eqref{ecu 1 prueba propo 4} implies
\[
f\in F(\widehat{\mu},\mu).
\]
Since $\widehat{\mu},\mu\in V_{IS}$, the Decomposition Lemma yields
\[
w=\mu(f)\in W(\mu,\widehat{\mu}),
\]
and hence
\[
\mu(w)P_w\widehat{\mu}(w).
\]
This contradicts the previous inequality.

\end{description}

All three cases lead to contradictions. Therefore, there is no
$V_{IS}\in V_{IS}(P)$ such that
\[
\{\mu\}\cup C(P)\subseteq V_{IS}.
\]
\end{proof}

\noindent\begin{proof}[Proof of Lemma~\ref{core es subconjunto de core de subperfil}]
Let $\widetilde{P}\sqsubseteq P$ and suppose that
$\mathcal{E}^{P\to\widetilde{P}}$ contains no stable pair of $P$.
Take any $\mu\in C(P)$.

First, every pair matched under $\mu$ is a stable pair of $P$ and therefore
is not eliminated in the construction of $\widetilde{P}$. Hence, $\mu$
remains individually rational under $\widetilde{P}$.

Suppose, by contradiction, that $\mu\notin C(\widetilde{P})$.
Then there exists a pair $(f,w)$ that blocks $\mu$ under
$\widetilde{P}$, so
\[
w\,\widetilde{P}_f\,\mu(f)
\qquad\text{and}\qquad
f\,\widetilde{P}_w\,\mu(w).
\]
Since $\widetilde{P}\sqsubseteq P$, the relative order of all remaining
acceptable partners is preserved. Therefore,
\[
wP_f\mu(f)
\qquad\text{and}\qquad
fP_w\mu(w),
\]
so $(f,w)$ also blocks $\mu$ under $P$. This contradicts
$\mu\in C(P)$.

Hence, $\mu\in C(\widetilde{P})$. Since $\mu\in C(P)$ was arbitrary,
\[
C(P)\subseteq C(\widetilde{P}).
\]
\end{proof}

\noindent\begin{proof}[Proof of Lemma \ref{todo core cumple estabilidad interna}]
    Let $\widetilde{P}$ be a preference profile such that $\widetilde{P} \sqsubseteq P$. Assume that the set $\mathcal{E}^{P\to \widetilde{P}}$ contains no stable pair of $P$, and that  $C(\widetilde{P}) \notin V_{IS}(P)$. This implies there are $\mu, \mu' \in C(\widetilde{P})$ such that $\mu \succ_S \mu'$ via $S$ for some coalition $S=\{f,w\}$ in $P$, and $\mu(f)=w$ also in $P$. Then, $\mu'\notin C(P).$ So, the fact that $\mu'\in C(\widetilde{P})\setminus C(P)$ implies that the pair $(f,w)$ is eliminated when constructing $\widetilde{P}$ from $P$. Hence, $\mu(f)\neq w$ in $\widetilde{P}$. This contradicts the assumption that that the set $\mathcal{E}^{P\to \widetilde{P}}$ contains no stable pair of $P$. Therefore, $C(\widetilde{P}) \in V_{IS}(P)$.
\end{proof}

\noindent\begin{proof}[Proof of Lemma~\ref{core de subperfil subconjunto de core de otro subperfil}]
Since $P_1\sqsubseteq P$ and
$\mathcal{E}^{P\to P_1}$ contains no stable pair of $P$,
Lemma~\ref{core es subconjunto de core de subperfil} implies that $C(P)\subseteq C(P_1).$

Similarly, since $P_2\sqsubseteq P_1$ and
$\mathcal{E}^{P_1\to P_2}$ contains no stable pair of $P_1$,
applying Lemma~\ref{core es subconjunto de core de subperfil} with
$P_1$ as the original profile yields $C(P_1)\subseteq C(P_2).$
Therefore,
\[
C(P)\subseteq C(P_1)\subseteq C(P_2).
\]
\end{proof}

\noindent\begin{proof}[Proof of Proposition  \ref{E_F y E_W cumple estabilidad interna}]
We begin by establishing internal stability for $\mathcal{M}^{E_F}$.

If the EADA mechanism stops at Round~0, then no interrupter pairs arise and the procedure coincides with the DA mechanism. 
Hence $\mu_{E_F}=\mu_F$, and the singleton $\{\mu_{E_F}\}$ trivially satisfies internal stability.

Consider now the case in which the EADA mechanism performs more than one round.  
Let $\mu_t$ and $\mu_r$ be the matchings obtained in Rounds $t$ and $r$, respectively, and assume without loss of generality that $r<t$.  
By the monotonicity properties of the EADA mechanism established by \cite{kesten2010school} (see Appendix~\ref{apendice descripcion EADAM}), we have 
\[
\mu_t \succeq_F \mu_r.
\]

Thus, no firm strictly prefers its assignment under $\mu_r$ to its
assignment under $\mu_t$. Consequently, $\mu_r$ cannot dominate
$\mu_t$ via any coalition. Thus, no firm strictly prefers its assignment under $\mu_r$ to its
assignment under $\mu_t$. Consequently, $\mu_r$ cannot dominate
$\mu_t$ via any coalition. That is, no coalition $S$ can satisfy $\mu_r \succ_S \mu_t$.

To complete the argument, suppose that there is a coalition \(S\) such that 
\[
\mu_t \succ_{S} \mu_r.
\]

Since we are in a one-to-one matching market, there is a pair $(f,w)$
matched under $\mu_t$, with $f,w\in S$, such that
\[
w=\mu_t(f), \qquad
wP_f\mu_r(f),
\qquad\text{and}\qquad
fP_w\mu_r(w).
\]

Because $(f,w)$ is matched under $\mu_t$, it cannot have been eliminated
in Round $r$ or in any subsequent round prior to $t$. Hence, $(f,w)$ is
mutually acceptable in the preference profile prevailing at Round $r$.
Since the relative order of the remaining acceptable partners is
preserved throughout the EADA reductions, the inequalities above imply
that $(f,w)$ blocks $\mu_r$ in the Round-$r$ profile.

However, \(\mu_r\) is the outcome of the DA mechanism applied to the profile induced in Round \(r\) of the EADA procedure. Therefore, \(\mu_r\) is a core matching for that profile and, in particular, cannot be blocked by any coalition. This yields a contradiction.

Consequently, no matching obtained in a later round can dominate a matching obtained in an earlier round.

Since no two matchings in $\mathcal{M}^{E_F}$ dominate each other, it follows that $\mathcal{M}^{E_F} \in V_{IS}(P)$. 
The same reasoning applies when workers propose, so $\mathcal{M}^{E_W} \in V_{IS}(P)$.
\end{proof}

\noindent\begin{proof}[Proof of Proposition \ref{EADAM subconjunto de P'}]
Let $\mu \in \mathcal{M}^{E_F}$. Then $\mu$ is the firm-optimal core
matching obtained by applying the DA mechanism in some round $t$ of the
EADA procedure. In particular, $\mu$ belongs to the core of the preference
profile prevailing in Round~$t$, and therefore no pair available at that
stage can form a coalition that dominates $\mu$.

In all subsequent rounds, the EADA mechanism only eliminates interrupter
pairs. Since eliminating pairs cannot create a new coalition capable of
dominating $\mu$, it follows that $\mu$ remains undominated throughout
the subsequent rounds.
Hence, $\mu \in C(SP^{\mu_{E_F},\mu_{E_W}})$, showing that
\[
\mathcal{M}^{E_F} \subseteq C(SP^{\mu_{E_F},\mu_{E_W}}).
\]
By an analogous argument, we also have 
\[
\mathcal{M}^{E_W} \subseteq C(SP^{\mu_{E_F},\mu_{E_W}}).
\]
\end{proof}

\noindent\begin{proof}[Proof of Proposition~\ref{proposition 1 de P*}]
By definition of $P^\star$,
\[
P^\star\sqsubseteq SP^{\mu_F,\mu_W}
\qquad\text{and}\qquad
P^\star\sqsubseteq SP^{\mu_{E_F},\mu_{E_W}}.
\]

Consider first any
$\mu\in C(SP^{\mu_F,\mu_W})$.
The only additional pairs eliminated when passing from
$SP^{\mu_F,\mu_W}$ to $P^\star$ are the EADA interrupter pairs.
By the discussion presented in Subsection \ref{subseccion de discusion de subperfiles}, these pairs cannot
support dominance against $\mu$.
Hence, $\mu$ remains undominated under $P^\star$, and therefore
\[
C(SP^{\mu_F,\mu_W})\subseteq C(P^\star).
\]

Similarly, let
$\mu\in C(SP^{\mu_{E_F},\mu_{E_W}})$.
The only additional pairs eliminated when passing from
$SP^{\mu_{E_F},\mu_{E_W}}$ to $P^\star$ are the pairs identified by
the Irving--Leather reduction.
By the discussion presented in Subsection \ref{subseccion de discusion de subperfiles}, these pairs cannot
support dominance against $\mu$.
Thus,
\[
C(SP^{\mu_{E_F},\mu_{E_W}})
\subseteq C(P^\star).
\]

Therefore,
\[
C(SP^{\mu_{E_F},\mu_{E_W}})
\subseteq C(P^\star)
\qquad\text{and}\qquad
C(SP^{\mu_F,\mu_W})
\subseteq C(P^\star).
\]
\end{proof}

\noindent\begin{proof}[Proof of Proposition \ref{core P^* satisfies internal stability}]
By the construction of $P^\star$, we have $P^\star \sqsubseteq P$ and 
\[
\mathcal{E}^{\,P \to P^\star}
    = 
\mathcal{E}^{\,P \to SP^{\mu_F,\mu_W}}
    \,\cup\,
\mathcal{E}^{\,P \to SP^{\mu_{E_F},\mu_{E_W}}}.
\]
Moreover, by Remarks~\ref{remark el conjunto de pares eliminados para muF y muW no contiene pares estables} 
and~\ref{remark el conjunto de pares eliminados no contiene pares estables},  
both 
\(\mathcal{E}^{\,P \to SP^{\mu_F,\mu_W}}\)  
and 
\(\mathcal{E}^{\,P \to SP^{\mu_{E_F},\mu_{E_W}}}\)  
contain no stable pairs of $P$.  
Therefore, the set 
\(\mathcal{E}^{\,P \to P^\star}\)  
also contains no stable pairs of $P$.  
Hence, by Lemma~\ref{todo core cumple estabilidad interna}, we conclude that  
\[
C(P^\star) \in V_{IS}(P).
\]
\end{proof}

The following technical lemma plays a key role in establishing both the
characterization of the vNM stable set and the
short-path-to-the-core result. Intuitively, it shows that whenever a
matching in $C(P^\star)\setminus C(P)$ is dominated by a matching outside
$C(P^\star)$, the latter is itself dominated by some matching in the
original core.

\begin{lemma}\label{Lema para unicidad}

Let \(\mu' \in C(P^\star) \setminus C(P)\).
Suppose that some matching \(\widetilde{\mu} \notin C(P^\star)\) dominates \(\mu'\) via a coalition \(S\).
Then there is a matching \(\mu \in C(P)\) and a coalition \(S'\) such that
\(\mu \succ_{S'} \widetilde{\mu}\).
\end{lemma}
\begin{proof}
Let $\mu'\in C(P^\star)\setminus C(P)$, and suppose that
$\widetilde{\mu}\notin C(P^\star)$ dominates $\mu'$ via a coalition
$S=\{f,w\}$:
\[
\widetilde{\mu}\succ_S\mu'.
\]
Since $\mu'\in C(P^\star)$, the pair $(f,w)$ cannot be mutually
acceptable under $P^\star$; otherwise, it would form a blocking coalition for $\mu'$ under
$P^\star$. Hence,
\[
(f,w)\in\mathcal E^{P\to P^\star}.
\]
By Definition~\ref{definicion de subperfil P star},
Remark~\ref{remark el conjunto de pares eliminados para muF y muW no contiene pares estables},
and Remark~\ref{remark el conjunto de pares eliminados no contiene pares estables},
the pair \((f,w)\) must be either an unstable pair between the two optimal matchings or an interrupter pair eliminated during the EADA procedure. By the construction of $P^\star$,
\[
\mathcal E^{P\to P^\star}
=
\mathcal E^{IL}
\cup
\mathcal E^{EADA},
\]
so we distinguish two cases.

\begin{description}
\item[ \textbf{Case 1: \((f,w)\in \mathcal{E}^{\,P\to SP^{\mu_F,\mu_W}}\).}]
Since $\widetilde{\mu}\succ_{\{f,w\}}\mu'$, we have
$\widetilde{\mu}(f)=w$. Moreover, by the characterization of the pairs
eliminated by the Irving--Leather reduction,
\[
\mu_F(f)P_fwP_f\mu_W(f).
\]
Hence,
\[
\mu_F(f)P_f\widetilde{\mu}(f)P_f\mu_W(f).
\]
Therefore, Proposition~\ref{Proposition Resultado 2} implies that
\[
\{\widetilde{\mu}\}\cup C(P)\notin V_{IS}(P).
\]
Since $C(P)\in V_{IS}(P)$, the failure of internal stability must involve
$\widetilde{\mu}$. Moreover, $\widetilde{\mu}$ cannot dominate a matching
in $C(P)$, since core matchings are undominated under $P$. Hence, there
is $\mu\in C(P)$ and a coalition $S'$ such that
\[
\mu\succ_{S'}\widetilde{\mu}.
\]

\item[\textbf{Case 2: \((f,w)\in \mathcal{E}^{\,P\to SP^{\mu_{E_F},\mu_{E_W}}}\).}]
Without loss of generality, suppose that \((f,w)\) is an interrupter pair
arising in the firm-proposing EADA mechanism.\footnote{The worker-proposing case is analogous.}
Then
\begin{equation}
\widetilde{\mu}(f)=w \; P_f \; \mu_F(f),
\label{eq:cond1}
\end{equation}
and
\begin{equation}
\mu_F(w)\; P_w\; \widetilde{\mu}(w)=f.
\label{eq:cond2}
\end{equation}

In this case, we claim that the dominating core matching is $\mu_F$ and the coalition $S'=\{\mu_F(w),\,w\}$, that is:
\[
\mu_F \succ_{\{\mu_F(w),\,w\}} \widetilde{\mu}.
\]

To prove this, consider the set \(F(\widetilde{\mu},\mu_F)\), and define
\[
\widetilde{\mu}\bigl(F(\widetilde{\mu},\mu_F)\bigr)
=
\{\,w\in W :
\exists x\in F(\widetilde{\mu},\mu_F)
\text{ such that }
\widetilde{\mu}(x)=w
\},
\]
and
\[
\mu_F\bigl(F(\widetilde{\mu},\mu_F)\bigr)
=
\{\,w\in W :
\exists x\in F(\widetilde{\mu},\mu_F)
\text{ such that }
\mu_F(x)=w
\}.
\]

From \eqref{eq:cond1}, we have
\(F(\widetilde{\mu},\mu_F)\neq\varnothing\).
Moreover, \eqref{eq:cond2} implies that
\[
w\in \widetilde{\mu}\bigl(F(\widetilde{\mu},\mu_F)\bigr)
\qquad\text{and}\qquad
w\notin \mu_F\bigl(F(\widetilde{\mu},\mu_F)\bigr).
\]

Let \(f'=\mu_F(w)\).
Since \(w\notin \mu_F(F(\widetilde{\mu},\mu_F))\), there is no
\(x\in F(\widetilde{\mu},\mu_F)\) such that \(\mu_F(x)=w\).
Therefore,
\(f'\notin F(\widetilde{\mu},\mu_F)\).
Because \(f'\) is matched to different workers under
\(\widetilde{\mu}\) and \(\mu_F\), it follows that
\(f'\in F(\mu_F,\widetilde{\mu})\), and hence
\[
\mu_F(f')\; P_{f'}\; \widetilde{\mu}(f').
\]

Combining this with \eqref{eq:cond2}, we obtain
\[
\mu_F \succ_{\{\mu_F(w),\,w\}} \widetilde{\mu},
\]
which proves the claim.
\end{description}
Since in both cases we have established that there is some \(\mu \in C(P)\) and some coalition $S'$ such that \(\mu \succ_{S'} \widetilde{\mu}\), the Lemma follows.
\end{proof}

\section{Appendix: The Irving--Leather reduction procedure and the set 
\texorpdfstring{$\mathcal E^{IL}$}{E-IL}}
\label{appendix:IL}

This appendix recalls the recursive reduction procedure introduced by
\cite{irving1986complexity} and presented in
\cite{roth1992two}. Our construction does not require the full reduction
procedure itself, but only the collection of pairs removed during its
execution. For completeness, we briefly describe the reduction and formally
define the corresponding set of eliminated pairs, denoted by
$\mathcal E^{IL}$.

The procedure starts from the original preference profile $P$ together with
the firm-optimal and worker-optimal core matchings,
$\mu_F$ and $\mu_W$.
More generally, suppose that two core matchings
$\mu^+$ and $\mu^-$ satisfy
$\mu^+\succeq_F\mu^-$.
The reduction associated with the interval
$[\mu^-,\mu^+]$ of the core lattice consists of the following three steps.

\begin{enumerate}

\item[\textbf{Step 1.}] (\textbf{Upper reduction})

For every firm $f$, remove from its preference list every worker strictly
preferred to $\mu^+(f)$.
Symmetrically, for every worker $w$, remove every firm that is strictly less
preferred than $\mu^+(w)$.

Consequently, no remaining pair can appear above $\mu^+$ in the lattice.

\item[\textbf{Step 2.}] (\textbf{Lower reduction})

For every firm $f$, remove every worker strictly less preferred than
$\mu^-(f)$.
Symmetrically, for every worker $w$, remove every firm strictly preferred to
$\mu^-(w)$.

After the first two steps, every acceptable pair belongs to the lattice
interval determined by $\mu^+$ and $\mu^-$.

\item[\textbf{Step 3.}] (\textbf{Mutual acceptability restoration})

The previous reductions may destroy mutual acceptability.
Whenever one agent still finds the other acceptable but the reverse is no
longer true, the pair is removed from the remaining preference list.
This elimination process is repeated until mutual acceptability is restored.

\end{enumerate}

The resulting reduced preference profile admits a collection of preference
cycles in the sense of \cite{irving1986complexity}.
Eliminating one such cycle produces a new core matching
$\mu^\sigma$ satisfying
\[
\mu^+\succ_F\mu^\sigma\succeq_F\mu^-.
\]

The same reduction procedure is then applied recursively to each interval
$[\mu^-,\mu^\sigma]$.
Repeating this construction recursively traverses the lattice of core
matchings until every stable matching has been generated.

Throughout the recursive execution of the procedure, Step~3 removes a
collection of pairs that fail to remain mutually acceptable after the
lattice-based restrictions.
These are precisely the pairs that we retain in our analysis.

\begin{definition}
The set
\[
\mathcal E^{IL}
\]
consists of every pair removed during Step~3 of the recursive
Irving--Leather reduction procedure.
Equivalently,
\[
(f,w)\in\mathcal E^{IL}
\]
if, at some recursive call, the lattice reductions performed in
Steps~1--2 make $(f,w)$ cease to be mutually acceptable, and the pair is
therefore eliminated during the subsequent mutual acceptability restoration.
\end{definition}

The first reduced preference profile considered in Section~\ref{subseccion subperfiles}
is obtained by removing exactly the pairs in
$\mathcal E^{IL}$.

\section{Appendix: Adapting the EADA mechanism to the one-to-one matching model}
\label{apendice descripcion EADAM}

The Efficiency-Adjusted Deferred Acceptance mechanism (EADA), introduced
by \cite{kesten2010school} in the school choice setting, is based on the
Deferred Acceptance mechanism (DA) of \cite{gale1962college}. The
procedure improves upon the DA outcome by successively identifying and
removing certain pairs, called interrupter pairs, and rerunning DA on the
resulting reduced preference profiles.\footnote{See
\cite{kesten2010school} for the original formulation of the mechanism in
the school choice setting.}

In this appendix, we adapt the EADA mechanism to the standard one-to-one
matching model. The adaptation is used only to identify the sequence of
matchings and interrupter pairs that enter the construction of the
EADA-induced subprofile introduced in
Subsection~\ref{subseccion subperfiles}.

The basic idea is the following. Starting from a DA outcome, the mechanism
identifies interrupter pairs generated during the execution of DA. The
agents on the proposing side consent to the removal of the corresponding
partners from their preference lists. DA is then applied again to the
resulting reduced profile, and the procedure is repeated until no further
interrupter pairs arise.

Unlike the original formulation in \cite{kesten2010school}, in which
students are the proposing side, in our one-to-one setting the procedure
is applied in both directions: once with firms proposing and once with
workers proposing. We denote the resulting matchings by $\mu_{E_F}$ and
$\mu_{E_W}$, respectively. Throughout the paper, we assume full consent,
so every interrupter pair identified by the procedure is removed.

We now introduce the formal notion of interrupter pairs underlying the mechanism.

\begin{definition}\label{par interruptor}
A firm $f$ is an \textbf{interrupter for worker} \(\boldsymbol{w}\) if:
\begin{enumerate}
    \item there is a Step $t$ and at least one other firm $f'$ such that $w$ accepts $f$ and rejects $f'$, and
    \item there is a later Step $t' > t$ and another firm $\widetilde f$ such that $w$ accepts $\widetilde f$ and rejects $f$.
\end{enumerate}
\end{definition}

When workers are the proposing side, the notion is defined symmetrically:
a worker may be an interrupter for a firm, and the corresponding pair is
again called an interrupter pair.

For a given execution of DA, let $\mathcal{IP}^t$ denote the set of
interrupter pairs generated at Step $t$.

Once the proposing side is fixed, the EADA mechanism proceeds in rounds.
At each round, DA is applied to the current preference profile. If no
interrupter pair arises during that execution, the procedure terminates.
Otherwise, we consider the last step at which interrupter pairs arise,
remove all such pairs from the preference lists of the proposing agents,
and run DA again on the resulting reduced profile. Under full consent,
every interrupter pair identified at that step is removed.

The formal description of the mechanism, when firms are the proposing side, is presented in Table~\ref{EADA-mechanism}.

\begin{table}[ht]
    \centering
    \begin{tabular}{@{}l p{13cm}@{}}\hline
        \textbf{Input:}  & A one-to-one matching market $P$. \\
        \textbf{Output:} & A matching $\mu_{E_F}$. \\[4pt]\hline\\[-6pt]
        
        \textbf{Round \boldmath$0$} 
        & \texttt{RUN} the DA mechanism to compute $\mu_F$. \texttt{DEFINE} $P^0 = P$. \\[8pt]
        
        \textbf{Round \boldmath$k \geq 1$}

        & \texttt{IF} $\mathcal{IP}^t = \emptyset$ for each step $t$ in the DA mechanism executed in round $k-1$ under $P^{k-1}$: \texttt{THEN, STOP}. \\[4pt]
        
        & \texttt{ELSE}:\\
         & \quad \texttt{FIND} the step $t$ in the DA mechanism executed in round $k-1$ under $P^{k-1}$ such that 
        $\mathcal{IP}^t \neq \emptyset$ and $\mathcal{IP}^{t'} = \emptyset$ for each $t' > t$. \\[4pt]
        & \quad  \texttt{CONSTRUCT} profile $P^k$ as follows: \\[2pt]
        & \quad \quad \texttt{IF} $(f,w)\in \mathcal{IP}^t$: $P_f^k \sqsubseteq P_f^{k-1}$ such that $(f,w)\in \mathcal{E}^{P_f^{k-1} \to P_f^{k}}$. \\
        & \quad \quad \quad   and $P_w^k= P_w^{k-1}.$\\
        &\quad \quad \texttt{ELSE}:\\
        &\quad \quad \quad  $P_f^k= P_f^{k-1}$ and $P_w^k= P_w^{k-1}.$; \\

        & \texttt{THEN, RUN} the DA mechanism for the profile $P^k$. \\[2pt]\hline
    \end{tabular}
    \caption{The modified $EADA$ mechanism}
    \label{EADA-mechanism}
\end{table}

The worker-proposing EADA mechanism is defined symmetrically. We denote
its final outcome by $\mu_{E_W}$.

Since at least one pair is eliminated whenever the procedure proceeds to
a new round and the market is finite, the mechanism terminates after
finitely many rounds. Moreover, if $P^k$ denotes the profile generated
in Round $k$, then
\[
P^k\sqsubseteq P^{k-1}\sqsubseteq\cdots\sqsubseteq P^1
\sqsubseteq P^0=P.
\]
For each execution of the mechanism, we record all interrupter pairs
eliminated throughout its rounds. As in
Subsection~\ref{subseccion subperfiles}, $\mathcal E^{EADA}$ denotes the
union of the pairs eliminated in the two executions of the mechanism,
one with firms proposing and the other with workers proposing.

\begin{example}
\textit{We execute the modified EADA mechanism (Algorithm \ref{EADA-mechanism}) for the market described in Example \ref{ex:wako}, with firms proposing.}

\begin{itemize}

\item \textit{\textbf{Round 0}:} \texttt{RUN} the DA mechanism under $P^0=P$ to compute $\mu_F$.

\begin{table}[ht]
    \centering
    \begin{tabular}{|c|c c c c|c| }
        \hline
        Step & $w_1$ & $w_2$ & $w_3$ & $w_4$ & Rejected \\
        \hline
        1 &  & $f_1$ & $f_2$,\cancel{$f_3$}& $f_4$ & $f_3$ \\
        2 &  & \cancel{$f_1$}, $f_3$ & $f_2$ & $f_4$ & $f_1$ \\
        3 &  & $f_3$ & $f_2$ & \cancel{$f_4$}, $f_1$ & $f_4$ \\
        4 &  & \cancel{$f_3$}, $f_4$ & $f_2$ & $f_1$ & $f_3$ \\
        5 & $f_3$ & $f_4$ & $f_2$ & $f_1$ & - \\
        \hline
    \end{tabular}
\end{table}

\item \textit{\textbf{Round 1}:} \texttt{FIND} the last step $t$ in the DA execution under $P^0$ such that $\mathcal{IP}^t \neq \emptyset$.  
In this case, $t=4$ and $\mathcal{IP}^4=\{(f_3,w_2)\}$.

\texttt{CONSTRUCT} the profile $P^1$ by removing $w_2$ from the preference list of $f_3$, while keeping all other preferences unchanged.

\begin{center}
\begin{tabular}{l l @{\hspace{1cm}} l l}
$P^1_{f_{1}}:$ & $w_{2}$, $w_{4}$, $w_{1}$, $w_3$ &
$P^1_{w_{1}}:$ & $f_{1}$, $f_{2}$, $f_{3}$, $f_{4}$ \\
$P^1_{f_{2}}:$ & $w_{3}$, $w_{1}$, $w_{2}$, $w_4$ &
$P^1_{w_{2}}:$ & $f_{2}$, $f_{4}$, $f_{3}$, $f_{1}$ \\
$P^1_{f_{3}}:$ & $w_{3}$, $w_{1}$, $w_{4}$ &
$P^1_{w_{3}}:$ & $f_{4}$, $f_{2}$, $f_{3}$, $f_{1}$ \\ 
$P^1_{f_{4}}:$ & $w_{4}$, $w_{2}$, $w_{3}$, $w_1$ &
$P^1_{w_{4}}:$ & $f_{3}$, $f_{1}$, $f_{4}$, $f_2$ \\ 
\end{tabular}
\end{center}

\texttt{THEN, RUN} the DA mechanism under $P^1$.

\begin{table}[ht]
    \centering
    \begin{tabular}{|c|c c c c|c| }
        \hline
        Step & $w_1$ & $w_2$ & $w_3$ & $w_4$ & Rejected \\
        \hline
        1 &  & $f_1$ & $f_2$,\cancel{$f_3$}& $f_4$ & $f_3$ \\
        2 & $f_3$ & $f_1$ & $f_2$ & $f_4$ & - \\
        \hline
    \end{tabular}
\end{table}

\item \textit{\textbf{Round 2}:} \texttt{FIND} the last step $t$ in the DA execution under $P^1$ such that $\mathcal{IP}^t \neq \emptyset$.  
Since no such step exists, we have $\mathcal{IP}^t=\emptyset$.

\texttt{THEN, STOP}. The algorithm outputs the matching
\[
\mu_{E_F} =
\begin{pmatrix}
f_1 & f_2 & f_3 & f_4 \\
w_2 & w_3 & w_1 & w_4  
\end{pmatrix}.
\]

\end{itemize}

\textit{When workers propose, the mechanism coincides with DA, and therefore $\mu_{E_W}=\mu_W$.
Consequently, across the two executions of the EADA mechanism, the only
interrupter pair eliminated in this example is $(f_3,w_2)$. Thus,
\[
\mathcal E^{EADA}=\{(f_3,w_2)\}.
\]}
\end{example}

We conclude by recording the properties of the EADA mechanism that will
be used in the paper. These properties follow from the corresponding
results in \cite{kesten2010school}, applied to the proposing side of the
one-to-one market.

First, the sequence of matchings generated by the mechanism is monotone
from the perspective of the proposing side. Suppose that firms propose
and let
\[
\mu_0=\mu_F,\mu_1,\ldots,\mu_K=\mu_{E_F}
\]
be the sequence of DA outcomes generated throughout the EADA procedure.
Then,
\[
\mu_{E_F}
=
\mu_K
\succeq_F
\mu_{K-1}
\succeq_F
\cdots
\succeq_F
\mu_1
\succeq_F
\mu_0
=
\mu_F.
\]
Thus, no firm becomes worse off as the procedure moves from one round to
the next.

Symmetrically, when workers propose, if
\[
\mu_0=\mu_W,\mu_1,\ldots,\mu_L=\mu_{E_W}
\]
denotes the corresponding sequence of DA outcomes, then
\[
\mu_{E_W}
=
\mu_L
\succeq_W
\mu_{L-1}
\succeq_W
\cdots
\succeq_W
\mu_1
\succeq_W
\mu_0
=
\mu_W.
\]

In particular, every matching generated by the firm-proposing EADA
procedure lies between $\mu_F$ and $\mu_{E_F}$ according to the firms'
preferences, whereas every matching generated by the worker-proposing
procedure lies between $\mu_W$ and $\mu_{E_W}$ according to the workers'
preferences. Hence, for every
$\mu\in\mathcal M^{E_F}$,
\[
\mu_{E_F}\succeq_F\mu\succeq_F\mu_F,
\]
and, for every
$\mu\in\mathcal M^{E_W}$,
\[
\mu_{E_W}\succeq_W\mu\succeq_W\mu_W.
\]

Second, under full consent, the final outcome of each EADA procedure is
Pareto efficient for the proposing side. Thus, there is no matching
$\mu'$ such that
\[
\mu'\succeq_F\mu_{E_F}
\]
with at least one firm strictly preferring its assignment under $\mu'$ to
its assignment under $\mu_{E_F}$. Symmetrically, there is no matching
$\mu'$ such that
\[
\mu'\succeq_W\mu_{E_W}
\]
with at least one worker strictly preferring its assignment under $\mu'$
to its assignment under $\mu_{E_W}$.

\end{document}

%% file: preamble.tex
\usepackage[utf8]{inputenc}
\usepackage[T1]{fontenc}
\usepackage{mathpazo}

\usepackage{amsmath,amssymb,amsfonts,amsthm}
\usepackage{mathrsfs}
\usepackage{nccmath}
\usepackage{cancel}

\usepackage{booktabs}
\usepackage{enumerate}
\usepackage{geometry}
\usepackage{graphicx}
\usepackage{soul}
\usepackage{emptypage}
\usepackage{natbib}
\usepackage{fancyhdr}

\usepackage{pgf,tikz}
\usetikzlibrary{intersections}
\usetikzlibrary{decorations.pathreplacing,angles,quotes}
\usetikzlibrary{arrows.meta}
\usetikzlibrary{arrows,decorations.markings}

\tikzset{
  myptr/.style={
    decoration={
      markings,
      mark=at position 1 with {\arrow[scale=2,>=stealth]{>}}
    },
    postaction={decorate}
  }
}

\usepackage{hyperref}
\hypersetup{
    colorlinks=true,
    citecolor=blue,
    filecolor=black,
    linkcolor=blue,
    urlcolor=blue
}

\newcommand*\samethanks[1][\value{footnote}]{\footnotemark[#1]}

\DeclareFontFamily{T1}{calligra}{}
\DeclareFontShape{T1}{calligra}{m}{n}{<-> s*[1.44] callig15}{}
\DeclareMathAlphabet{\mathcalligra}{T1}{calligra}{m}{n}

\newtheorem{theorem}{Theorem}

\newtheorem{corollary}{Corollary}

\newtheorem{definition}{Definition}

\newtheorem{example}{Example}

\newenvironment{examplecont}[1]
{%
  \addtocounter{example}{-1}%
  \begin{example}%
}
{%
  \end{example}%
}

\newtheorem{lemma}{Lemma}

\newtheorem{proposition}{Proposition}
\newtheorem{remark}{Remark}

\renewenvironment{proof}[1][Proof]{\noindent \emph{#1.} }{\hfill$\square$\vspace{5pt}}
